\documentclass[journal]{IEEEtran}
\usepackage{amsmath,amsthm}
\usepackage{amsfonts,amssymb}
\usepackage{algorithmic}
\usepackage{bm}
\usepackage{graphicx}
\usepackage{cite}
\usepackage{relsize}
\usepackage{bbm}
\usepackage{algorithm,algorithmic}
\usepackage{booktabs}

\newtheorem{theorem}{Theorem}
\usepackage{caption}
\usepackage{subcaption}
\ifCLASSINFOpdf
\else
\fi
\begin{document}
% Linebreaks \\ can be used within to get better formatting as desired. Deterministic Maximum Likelihood
\title{Robust Expectation-Maximization Algorithm for  DOA Estimation of Acoustic Sources in the Spherical Harmonic Domain}

\author{Hossein Lolaee,~\IEEEmembership{~Student~Member,~IEEE,} 
	{Mohammad~Ali~Akhaee*,~\IEEEmembership{~Member,~IEEE,}}
	
	\thanks{H. Lolaee and M.A. Akhaee are with the Department of Electrical and Computer Engineering, College of Eng., University of Tehran,, Tehran, Iran, 14588\-89694 (e-mail:  {hossein\_lolaee, akhaee}@ut.ac.ir)}
}
\maketitle

\begin{abstract}
The direction of arrival (DOA) estimation of sound sources has been a popular signal processing research topic due to its widespread applications. Using spherical microphone arrays (SMA),  DOA estimation can be applied in the spherical harmonic (SH) domain without any spatial 	ambiguity.
% Spatial symmetry helps the spherical arrays to capture the three-dimensional information of sound sources.
	% and effectively estimate DOAs.
	However, the environment reverberation and noise can
	degrade the estimation performance.
 In this paper, 	we propose a new  expectation maximization (EM) algorithm  for deterministic maximum likelihood (ML)  DOA  estimation 
  of $L$ sound sources in the presence of spatially
 nonuniform noise in the SH domain.
 % DOA estimation  under deterministic assumption of sources signal in SH domain. 
 Furthermore a  new closed-form Cramer-Rao bound (CRB) for the deterministic ML DOA estimation  is derived for the signal model in the SH domain.
 The main idea of the proposed algorithm is
 considering the general model of the received signal in the
 SH domain, we reduce the complexity of the ML estimation by breaking it down into two steps: expectation and maximization steps.
 The proposed algorithm reduces the complexity from  
 $2L$-dimensional space   to $L$ $2$-dimensional space.
 Simulation results indicate that
 the proposed algorithm shows at least an improvement of 6dB
 in robustness in terms of root mean square error (RMSE). Moreover, the RMSE of the proposed algorithm 
 is very close to the CRB compared to the recent methods in  reverberant and noisy
 environments  in the large range of signal to noise ratio.

%  Through different scenarios, we show the proposed method more robust in noisy and reverberant environment than recent methods in large range of SNR. 
%	Simulation results indicate that the proposed method 
%	 In this paper, using the exact model of the received signal in SH domain  including the  environment noise, a novel  maximum-likelihood (ML)-based algorithm is presented to estimate  the directions of multiple sound sources.  
%	Its implementation is based on an iterative procedure  under stochastic assumption of sources signals.
%	Also the Cramer-Rao bound (CRB) has been provided for this model. Simulation results indicate that the proposed method exhibit better robustness in the reverberant and noisy environments.         
\end{abstract}
%

\begin{IEEEkeywords}
	direction of arrival estimation, spherical microphone array, spherical harmonics
\end{IEEEkeywords}

% \fi
\IEEEpeerreviewmaketitle

\section{Introduction}
The direction of arrival (DOA)  estimation of sound sources has been a popular signal processing research topic due to its widespread applications, including speech enhancement and dereverberation, robot auditory, and spatial room acoustic analysis and synthesis.  Various algorithms and array structures have been proposed so far for different applications. Among different types of arrays 
%the dominant shape of the arrays 
including spherical, circular and linear, spherical arrays have attracted more attention recently. The spatial symmetry of spherical arrays helps us to capture the 3-D information of sound sources without spatial ambiguity.
Moreover, using spherical arrays the sound field can be analyzed by an orthonormal basis in the spherical harmonic (SH) domain.
The main advantage of analysis in the SH domain is the decoupling of frequency-dependent and angular-dependent components \cite{Teutsch_2008}. Better compression of spatial information, wide-band array beamforming, and linear analysis of array output signals  \cite{Teutsch} are the other advantages of  processing in the SH domain.

% The mike signals transferred to the SH domain is known as higher-order ambisonic  (HOA) signals. 

The traditional DOA estimation methods can be divided into three categories: time-delay, beamforming, and subspace based methods;  In the first category, the DOA is estimated using the time delay between the received signals in the microphone pairs of the array  \cite{Blandin_2012}. In the second category, the direction corresponding to the highest beamformer power is declared as the source direction \cite{Tung}. The third group is known by the famous multiple signal classification (MUSIC) algorithm \cite{Schmidt_1986}. Estimation of the signal parameters via rotational invariance techniques (ESPRIT) is another notable method within this category \cite{Roy_1989}. Various DOA estimation algorithms have been developed based on these three  categories in the SH domain\cite{Goossens_2009,Sun_2011,Li_2011,Kumar_2016,Epain_2012,Noohi_2013,Noohi_2015,Jin_2014,Tervo_2015,Moore_2016,Pan_2016,Haohai_Sun_2011}. 

The sound source reverberation causes producing correlated and coherent acoustic signals  which affects the performance of traditional methods specially spectral based ones \cite{Krim_1996}.
 Also, due to the rank reduction of  the spatial covariance matrix in the  reverberant environment, the  MUSIC and ESPRIT algorithms suffer from performance degradation. Although in \cite{Goossens_2009,Sun_2011,Li_2011,Kumar_2016}, MUSIC and ESPRIT are applied in the SH domain,  they lost accuracy in high reverberation.
% but in high reverberation still show weakness. 

In the series of \cite{Epain_2012}, \cite{Noohi_2013} and \cite{Noohi_2015}, the DOA estimation method is proposed based on the independent component analysis (ICA) by using directional sparsity of sound sources. In \cite{Epain_2012}, the unmixing matrix is extracted by applying the ICA model to the SH domain signals; Then DOA is estimated by comparing its columns  with the dictionary of possible plane-wave source directions steering vectors.
Since this method suffers from low resolution, by combining ICA and sparse recovery methods its performance can be improved \cite{Noohi_2013}.
% This method suffers from low resolution and in \cite{Noohi_2013} was improved by combining ICA and sparse recovery methods. 
 In \cite{Noohi_2015}, the authors improve the convergence of solvers to a local minimum by using spatial location of the sound sources as a primary information.   In all of these  methods, the proper estimation is achievable in special scenarios. Also,  DOA estimation strongly degrades in the  reverberant and noisy environment. 

%\vspace{-1px}
 The linear signal model in the SH domain and capturing 3-D information of sources without spatial ambiguity using spherical microphone arrays (SMA) motivated us to estimate DOAs  in the SH domain. Considering the general model of the received signal in the SH domain, we propose a new  expectation maximization (EM) algorithm  for deterministic maximum likelihood (ML)  DOA  estimation 
for $L$ sound sources in the presence of spatially
nonuniform noise.
Furthermore,  a  new closed-form Cramer-Rao bound (CRB) for the deterministic ML DOA estimation  is derived for the signal model in the SH domain.
The ML estimator requires an exhaustive search in a $2L$-dimensional space. In order to reduce
the complexity, we break down the ML estimation to $L$ $2$-dimensional space. 
 Simulation results indicate that
the proposed algorithm is robust in  the reverberant and
noisy environments in the large range of signal to noise ratios (SNRs).
Based on the simulations, the proposed method
can provide   improvement in the robustness
%6 and 7 dB in robustness 
%comparing to record results for MUSIC and ICA, respectively, in terms of   root mean square error (RMSE)
  in   reverberant and
noisy environments.  
%Moreover,  the RMSE of the proposed algorithm 
%is very close to the CRB compared to the mentioned methods  in the large range of SNRs. 
%All of these observations demonstrate the efficiency of our method.  

% Through different scenarios, we show the proposed method more robust in the noisy and reverberant environment than state of the art  methods for the large range of SNR.  Moreover the
%closed-form expression of Cramer-Rao Bound (CRB) for signal model in SH domain is derived. Simulation results indicate that the proposed method is very close to CRB compared to other methods which demonstrates the efficiency of our method.  
%
%use the exact model of the received HOA signals and consider the environment noise to estimate DOAs of multiple sound sources. A novel iterative method is presented based on ML estimation of sources direction under stochastic assumption of sources signals.
% Moreover the
%closed-form expression of CRB for this model
%is derived.
%  Due to the high dependency of ML estimation of sources direction to the unknown noise variances, an exhaustive search in entire parameter space is required. Therefore, the iterative approach for ML estimation is proposed to decrease the complexity.
%Simulation results indicate
%that  our approach compared to the state of the art methods  is significantly closer to the CRB
%and show better robustness in the reverberant and noisy environments.  

%(This paragraph is not required for conference papers)
%\vspace{-1px}
The remainder of this paper is organized as follows. In section II, 
 the signal model in the SH domain is  investigated.
% reviews the  signal model in the SH domain.
 Section III indicates the proposed EM algorithm for the  ML DOA estimation.
 The derivation of  the deterministic CRB of the 
  signal model is presented  in detail in section IV. The evaluation and comparison of the proposed algorithm with other methods through different scenarios is reported in section V. Finally,  section VI concludes the paper.
%  the conclusion is provided in part VI.
\section{Signal Model}
In this section, a model for the received signal in  the SH domain is presented using the approach provided in \cite{Kumar_2015}.
%In this section, we model the received signal in the SH domain,  based on the works in \cite{Kumar_2015}.
 Consider the spherical array of $I$ identical omnidirectional microphones and the  $i$'th microphone located at  Cartesian coordinates  of ${\bm r}_i = [r \sin \theta_i\cos \phi_i,r \sin \theta_i\sin \phi_i,r \cos \theta_i]^T$, where $(r, \theta_i, \phi_i)$ denote the corresponding spherical coordinates. 
The notations describing the spherical geometry  are illustrated in Fig. \ref{fig:notation}.
\begin{figure}
	\centering
	\includegraphics[width=0.7\linewidth]{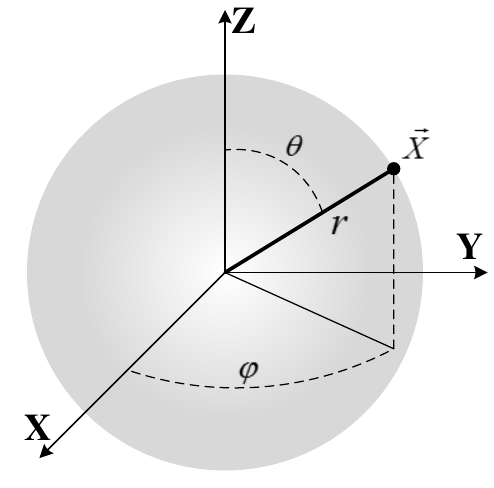}
	\caption{The notations describing the spherical geometry}
	\label{fig:notation}
\end{figure}
Assume that there exist $L$  plane-wave source signals where $l$'th source  impinging  in the angular direction 
 ${\bm \Psi}_l \stackrel{\Delta}{=} (\theta'_l,\phi'_l)$
with wave-number $k$. 
The received signal at the $i$'th microphone from the $l$'th source at time $t$ is
$ s_l\left(t-\tau_i(\mathbf{\Psi}_l)\right) $, where $\tau_i(\mathbf{\Psi}_l)$ is the  propagation delay of the $l$'th  source between the reference and $i$'th microphone.
For a narrow band sound source, the received signal can be written in this form:
\begin{equation}
s_l(t-\tau_i(\mathbf{\Psi}_l)) = e^{-j\bm{k}_l^T \bm{r}_i} s_l(t),
\end{equation}
%Considering  narrow-band  assumption of the sound source,  we have $s_l(t-\tau_i(\mathbf{\Psi}_l)) = e^{-j\bm{k}_l^T \bm{r}_i} s_l(t)$,
 where  
$\bm {k}_l =- [k \sin\theta'_l \cos\phi'_l,$ $ k \sin\theta'_l \sin\phi'_l, k \cos\theta'_l]^T $ is the wave-vector corresponding to the $l$'th plane-wave. 
% The $i$'th microphone observation at time $t$ is modeled as \cite{Krim_1996}
The received signal at $i$'th microphone at time $t$ is \cite{Krim_1996}:
\begin{equation}
x_i(t) = \sum_{l=1}^L e^{-j\bm {k}_l^T \bm {r}_i} s_l(t) + n_i(t),\quad 1\leq i\leq I,
%\quad t=1,2,\ldots,N_s,
\label{equ:InsPress2}
\end{equation}
where  
%$\bm {k}_l =- [k \sin\theta_l \cos\phi_l,$ $ k \sin\theta_l \sin\phi_l, k \cos\theta_l]^T $ is the wave-vector corresponding to the $l$'th plane-wave and
 $n_i(t)$ is the additive white  Gaussian noise with zero mean and variance 
 $\sigma^2$  of the $i$'th microphone. The received signal in  (\ref{equ:InsPress2}) can be restated in a matrix form as:
\begin{equation}
{\bm x}(t) = {\bf A}({\bm \Psi }){\bm  s}(t) + {\bm  n}(t),
%,\qquad t=1,2,\ldots,N_s,
\label{equ:VecInsPress}
\end{equation}
where
$
{\bm x}(t)\stackrel{\Delta}{=}
\begin{bmatrix}
{x}_1(t),  
{x}_2(t),
\ldots,
{x}_I(t)
\end{bmatrix}^T
\label{equ:PresInsVec}
$,
$\bm {s}(t) \stackrel{\Delta}{=} [	s_1(t) ,$ $	s_2(t) ,\ldots,	s_L(t) ]^T$,
$\bm {n}(t) \stackrel{\Delta}{=} [	n_1(t) ,	n_2(t) ,\ldots,	n_I(t) ]^T$,
$\mathbf{\Psi} \stackrel{\Delta}{=} \{\mathbf{\Psi}_l,l=1,\ldots,L\}$
and 
 ${\bf A}$ is the  $I \times L$ direction matrix  which is composed of the signal direction vectors as: 
\begin{equation}
{\bf A}({\bf \Psi }) \stackrel{\Delta}{=} 
\begin{bmatrix}
{\bm a}({\bf \Psi }_1),{\bm a}({\bf \Psi }_2),\ldots,{\bm a}({\bf \Psi }_L)
\end{bmatrix}, 
\end{equation}
where
\begin{equation}
{\bm a}({\bf \Psi }_l) \stackrel{\Delta}{=} 
\begin{bmatrix}
e^{-j\mathbf{k}_l^T \mathbf{r}_1},e^{-j\mathbf{k}_l^T \mathbf{r}_2},\ldots,e^{-j\mathbf{k}_l^T \mathbf{r}_I}
\end{bmatrix}^T.
\label{equ:SteeringMat}
\end{equation}
The $i$'th element of the ${\bm a}({\bf \Psi }_l)$ is 
 the incident sound field to the $i$'th microphone of the array from the $l$'th  unit amplitude plane-wave. On the other hand, by solving the  wave equation in the spherical coordinates \cite{Williams_2000}, the following equality can be obtained:
 \begin{equation}
 e^{-j\bm {k}_l^T \bm {r}_i} = \sum_{n=0}^{\infty} \sum_{m=-n}^{n} b_n(kr)  Y_n^m({\bf \Psi}_l) Y_n^m({\bf \Phi}_i),
 \label{equ:PlwExpan}
 \end{equation}
%sound wave pressure received at array from the direction $(\theta_i,\phi_i)$  at location $(r,\theta,\phi)$ on the array will be:
%\begin{equation}
%p(r,\theta,\phi)=\sum_{n=0}^{\infty} \sum_{m=-n}^{n}  b_n(kr) Y_n^m(\theta,\phi) Y_n^m(\theta_i,\phi_i)^*
%\label{equ:PresPlanwa}
%\end{equation}
where $b_n(kr)$  is the  mode strength of order $n$ and it is defined  for open sphere as:
    \begin{equation} 
b_n(kr) \stackrel{\Delta}{=} 
4\pi j^n j_n(kr), 
\label{equ:ModAmpDef}
\end{equation}
%    \begin{equation} 
%b_n(kr) \stackrel{\Delta}{=} \left\{
%\begin{array}{ll}
%\hspace{-5pt}
%4\pi j^n j_n(kr) & \text{{open sphere}} \\ 
%\hspace{-5pt}
%4\pi  j^n \Big(j_n(kr)-\frac{j^{'}_n(kr_0)}{h^{'}_n(kr_0)}h_n(kr)\Big) &  \text{{rigid sphere,}}
%\end{array} \right.
%\label{equ:ModAmpDef}
%\end{equation}
  ${\bm \Phi}_i \stackrel{\Delta}{=} (\theta_i,\phi_i)$ and $Y_n^m(\cdot)$  is the real-valued  spherical harmonic  of order $n$ and degree $m$ defined as:
       \begin{align}
    Y_l^m({\theta,\phi}) \stackrel{}{=}& \sqrt {{{2l + 1} \over {4\pi}}{{(l - m)!} \over {(l + m)!}}} \;P_l^{\vert m \vert}({\cos \theta}) \nonumber \\ 
    &\qquad\quad\times\left\{
    \begin{matrix}
    \cos(m\phi) &  {\rm for }\, m \geq 0 \\
    \sin(\vert m\vert\phi) & {\rm  for }\, m < 0 
    \end{matrix}
    \right.      \label{equ:SpherHarm}.
    \end{align}
%    \begin{equation}
%    Y_n^m(\theta,\phi) \stackrel{\Delta}{=} \sqrt{\frac{2n+1}{4\pi}\frac{(n-m)!}{(n+m)!}} P_n^m(\cos\theta) e^{jm\phi}
%    \label{equ:SpherHarm},
%    \end{equation}
In (\ref{equ:ModAmpDef}) and (\ref{equ:SpherHarm}), $j = \sqrt{-1}$, $j_n$ and $h_n$  are the spherical Bessel and Henkel functions, $r_0\leq r$ is the radius of the rigid sphere and  $P_n^m$ is the associated Legendre polynomial of order $n$ and degree $m$ \cite{Driscoll_1994,Whittaker_1996}.
% $i$'th entry of \ref{equ:SteeringMat}  refers to the sound field pressure of the $l$'th plane wave source at location  $\bf r_i$. Therefore, according to (\ref{equ:PresPlanwa}) we have  \cite{zotter2009analysis}:
%\begin{equation}
%e^{-j\mathbf{k}_l^T \mathbf{r}_i} = \sum_{n=0}^{\infty} \sum_{m=-n}^{n} b_n(kr)  Y_n^m({\bf \Psi}_l)^* Y_n^m({\bf \Phi}_i).
%\label{equ:PlwExpan}
%\end{equation}
Applying a proper truncation order $N$ \cite{Sun_2012} to (\ref{equ:PlwExpan}), the sound field can be approximated inside a sphere of radius $\hat{r}$   centered at the  origin, as follows \cite{Epain_2012}:
%\cite{Kennedy2007}
%Including the first $N$  order of spherical harmonic, $N$-order series expansion is derived. Resulting series is a proper estimation of the sound field  . Thus we can say:
\begin{equation}
e^{-j\bm {k}_l^T \bm {r}_i} = \sum_{n=0}^{N} \sum_{m=-n}^{n} b_n(kr)  Y_n^m({\bf \Psi}_l) Y_n^m({\bf \Phi}_i), \quad \parallel {\bm r_i} \parallel \leq \hat{r},
\label{equ:PlwExpanTrunc}
\end{equation}
where $\hat{r}={2N\over ek_l}$ and $e$ is       the Euler number. Rewriting (\ref{equ:PlwExpanTrunc})
in the matrix form, we have:
\begin{equation}
{\bf A}({\bf\Psi}) = {\bf Y}({\bf\Phi }){\bf B}(kr){{\bf Y}}({\bf\Psi })^T,
\label{equ:DecompInsPress}
\end{equation}
where $\mathbf{\Phi} \stackrel{\Delta}{=} \{\mathbf{\Phi}_i,i=1,\ldots,I\}$ and
${\bf Y}({\bf\Psi })$ is the source spherical harmonics matrix of size $L \times (N+1)^2$ and defined as:
\begin{equation}
{\bf Y}({\bf\Psi }) \stackrel{\Delta}{=}  [{\bm  y}({\bf\Psi }_1)^T,{\bm y}({\bf\Psi }_2)^T,\ldots,{\bm y}({\bf\Psi }_L)^T]^T 
\label{equ:SpheHarmMat}
\end{equation}
where
\begin{equation*}
{\bm y}({\bf\Psi }_l) = [Y_0^0({\bf\Psi }_l),Y_1^{-1}({\bf\Psi }_l),Y_1^0({\bf\Psi }_l),Y_1^1({\bf\Psi }_l),\ldots,Y_N^N({\bf\Psi }_l)],
\label{equ:SpheHarmVec}
\end{equation*}
the array spherical harmonics matrix, ${{\bf Y}}({\bf\Phi })$,  is the size of $I\times (N+1)^2$ and defined similar to (\ref{equ:SpheHarmMat}) and the mode strength matrix, ${\bf B}(kr)$, is the size of  $(N+1)^2 \times (N+1)^2$ and defined as follows:
\begin{equation*}
{\bf B}(kr) \stackrel{\Delta}{=} \text{diag} \big\{{b_0}(kr),{b_1}(kr),{b_1}(kr),{b_1}(kr), \ldots,{b_N}(kr) \big\}.
\end{equation*}

The spherical harmonics decomposition of the received signal  can be obtained as\cite{Driscoll_1994,Williams_2000}:
%{x_{nm}}(t) &= \int_{0}^{2\pi} \int_{0}^{\pi}  x(t) Y_n^m(\theta,\phi)^{\ast}\sin\theta \, d\theta d\phi \\ and $\alpha_i$ are
\begin{equation}
{x_{n,m}}(t) = \int_{\Omega\in S^2}  x(t) Y_n^m(\Omega) \, d\Omega  
\label{equ:xnmint}
\end{equation}
where $\Omega = (\theta,\phi)$ and 
${x_{n,m}}(t)$ are  the coefficients of the spherical harmonics decomposition and $S^2$ denotes the entire surface
area of the unit sphere. 
Since the number of microphones is limited, obtaining ${x_{n,m}}(t)$ using   \eqref{equ:xnmint} is not applicable;
%Obtaining ${x_{n,m}}(t)$ using   \eqref{equ:xnmint} is not applicable.  Because the number of microphones is limited,   
We do not access  $ x(t)$ on the entire surface of the array. Actually the
spherical  microphone array (SMA) perform spatial sampling of $ x(t)$ using real-valued sampling   weights, $\alpha_i$,  corresponding to the $i$'th microphone \cite{Chen_2015,Rafaely_2005}:
\begin{equation}
{x_{n,m}}(t) \cong \sum\limits_{i = 1}^I {\alpha_i}{x_i}(t) {Y_n^m}({{\bf \Phi}_i}),
\label{equ:ComSpherWei}
\end{equation}
%$\alpha_i$ is the real-valued 
%weighting parameter corresponding to the $i$'th microphone and  obtained according to the spatial sampling scheme \cite{Chen_2015}.  
%$\alpha_i$ are chosen such that the series is a proper estimation of  ${p_{n,m}}({\bf \Psi};t)$  \cite{Rafaely2005}.
Equation (\ref{equ:ComSpherWei}) can be represented as in a matrix form
% Restating (\ref{equ:ComSpherWei}) in the  matrix form, we have:
\begin{equation}
{{\bm  x}_{{\bf nm}}}(t) = {{\bf Y}}({\bf \Phi})^T{\bf \Sigma} {\bm x}(t),
\label{equ:VecCompSphe}
\end{equation}
where
% ${\bf \Sigma}	\stackrel{\Delta}{=} \text{diag} \big\{\alpha_1,\alpha_2,\ldots,\alpha_I \big\}$ and 
% ${{\bm  x}_{{\bf nm}}}(t)$ 
%% and $\bf \Sigma$ and 
%   is defined as
\begin{align*}
{\bf \Sigma}	&\stackrel{\Delta}{=} \text{diag} \big\{\alpha_1,\alpha_2,\ldots,\alpha_I \big\} \text{ and }\\
{{\bm  x}_{{\bf nm}}}(t)& \stackrel{\Delta}{=}[x_{0,0}(t),x_{1,-1}(t),x_{1,0}(t),x_{1,1}(t),\ldots,x_{N,N}(t)]^T.
\end{align*}
%\begin{equation}
%{\bf \Sigma}	\stackrel{\Delta}{=} \text{diag} \big\{\alpha_1,\alpha_2,\ldots,\alpha_I \big\}.
%\end{equation}
 Considering  (\ref{equ:ComSpherWei}), the spherical
 harmonics are  orthonormal as represented in \cite{Rafaely_2005}
 % perpendicularity feature of spherical harmonics in the discrete state follows the below equation  \cite{arfken2012mathematical}:
 \begin{equation}
 {{\bf Y}}({\bf \Phi })^T{\bf \Sigma {\bf Y}}({\bf \Phi }) = {\bf I},
 \label{equ:OrthDiscr}
 \end{equation}
 where $\bf I$ is ${(N+1)^2}\times {(N+1)^2}$ identity matrix. 
 Replacing (\ref{equ:DecompInsPress}) in (\ref{equ:VecInsPress}) and  multiplying both sides of the  equation by  ${\bf \Gamma} \stackrel{\vartriangle}{=}{\bf B}^{-1}(kr)	{{\bf Y}^T}({\bf \Phi }){\bf \Sigma}$ and using equations (\ref{equ:VecCompSphe}) and (\ref{equ:OrthDiscr}), the     received signal model in the SH domain can be calculated as:
 %_{{\bf nm}}}
 \begin{equation}
 {\bm b}(t) = {{\bf Y}}({\bf \Psi })^T{\bm s}(t) + {\bm z}(t),\quad t=1,\ldots,N_s \label{equ:DataModel1} 
 \end{equation}
 where 
 $
 {\bm b}(t) \stackrel{\Delta}{=} {\bf \Gamma} {\bm x}(t)
 $,
  $
 {\bm z}(t) \stackrel{\Delta}{=} {\bf \Gamma} {\bm n}(t) $ and $N_s$ is the number of snapshots.
 ${\bm b}(t)$ is named  higher-order ambisonic (HOA) signal  \cite{Epain_2012} and   
  ${\bm z}(t)$   is the noise vector in the SH domain.
  According to \eqref{equ:DataModel1},  the HOA signal is a linear instantaneous mixture of the sources signal. 
  %   ${\bm b}(t)$ and ${\bm z}(t)$ are obtained as follows:
% \begin{align}
% {\bm z}(t) = {\bf \Gamma} {\bm n}(t), \label{equ:NoiseCov}\\
% {\bm b}(t) = {\bf \Gamma} {\bm x}(t).
% \label{equ:HOAVec}
% \end{align}
% Note that matrix  ${\bf\Gamma} $   is known for the mike array at hand. In literature, time domain signal ${{\bf b}_{{\bf nm}}}$  is called high order audio (HOA) signal.
Transforming the received  signals at the array to the SH domain is performed by applying the time domain encoding filter, ${\bf\Gamma}$, to the received signal of the  array \cite{Sun_2012}. It must be noticed that the transformation filter, ${\bf\Gamma}$, is known for the given array. 
%  Further information about signal transformation from the mike domain to the spherical harmonics domain can be found in \cite{Epain2012}.
%The advantages of signal processing in
%the SH domain are linear mixture model of the microphone
%signals \cite{Epain_2012}.
   
%   In this paper, we utilize real-valued spherical harmonic \cite{Whittaker_1996}:
%   \begin{align}
%   Y_l^m({\theta,\phi}) \stackrel{}{=}& \sqrt {{{2l + 1} \over {4\pi}}{{(l - m)!} \over {(l + m)!}}} \;P_l^{\vert m \vert}({\cos \theta}) \nonumber \\ 
%    &\qquad\quad\times\left\{
%   \begin{matrix}
%   \cos(m\phi) &  {\rm for }\, m \geq 0 \\
%   \sin(\vert m\vert\phi) & {\rm  for }\, m < 0 
%   \end{matrix}
%   \right. ,
%   \end{align}
%    where $P_l^m$  is the order $l$, degree $m$ associated Legendre polynomial.   
%    Therefore,  $T$   will be used in (\ref{equ:DataModel1}) instead of  $H$.  
      
   %%%%%%%%%%%%%%%%%%%%%%%%%%%%%%%%%%%
   \section{Deterministic Model for DOA Estimation}
   In this section, a new ML DOA estimation based on EM algorithm in the SH domain  is proposed by considering unknown and deterministic sources.
%   Here,   considering  unknown and deterministic sources,
%   assuming the sources signal are unknown  and deterministic, 
%   we propose new ML  DOA  estimation 
%   based on EM algorithm.
%    So, we should estimate both source signal and its DOA.
It is worth noting that the additive noise in 
(\ref{equ:DataModel1}) is spatially  nonuniform white noise. 
%    According to (\ref{equ:DataModel1}), the
%    uniform white noise assumption is not  realistic. 
Because the $ {\bf \Gamma}$ filter applied to $ {\bm n}(t) $ is not an identity matrix.
%    Because the noise vector z(t) obtained by filtering n(t) with Γ and the Γ filter is not an identity matrix, therefore z(t) is a nonuniform noise. 
        In the following, we investigate two cases of assumption for DOA estimation in the SH domain: i)
    uniform and ii) nonuniform noise.
%     across the microphones.
%     and In this regard,    we consider two cases:  independent elements of noise vector    with identical and nonidentical noise variances.
    
   \subsection{Uniform Noise }
      \label{subsec:uniform}
   In this subsection, the deterministic ML DOA estimation for uniform noise case \cite{Chen_2002} is reviewed. The important formulas which have been employed in the subsequent sections are discussed.
%       we review the deterministic ML DOA estimation under uniform noise assumption   according to   \cite{Chen_2002}.
%        Here, we mention  substantial formula  exploited in the subsequent sections.
   %   Although different channels have noise vectors with unequal  variances based on (\ref{equ:NoiseCov}), here we assume that noise variance is the same at different channels.
   
       Suppose that  the additive  noise vector  to be zero mean Gaussian with the covariance matrix of 
   ${\bf R}_{n}=\sigma^2\mathbf{I}$. The set of
    unknown parameters are defined as 
    ${\bf \Omega} \stackrel{\Delta}{=} \{{\bf \Psi},{\bf  S}, \sigma^2\}$, where
%   \begin{equation}
%   ${\bf\Theta} = [\Theta_1,\ldots,\Theta_{2L}]^T 
%   = [\bm{\theta}^{T},\bm{ \phi}^{T}]^{T} = [\theta_1,\ldots,\theta_L,\phi_1,\ldots,\phi_L]^T $ and 
%   \label{equ:Param}
%   \end{equation}
%   \begin{equation}
$   {\bf S} \stackrel{\Delta}{=} \{{\bm  s(}1),\ldots,{\bm  s(}N_s)\}$.
%   \end{equation}
Thus, the likelihood function of the received signal in the SH domain will be:
\begin{align}
	{ f({\bm b};{\bf \Omega})}=&{1 \over {{(2\pi \sigma^2) ^{P N_s/2}} }}\nonumber\\ &\times\exp\left(-{1\over 2  \sigma^2} \sum_{t=1}^{N_s}\parallel{\bm  b}(t)-{{\bf Y}}({\bf \Psi })^T{\bm s}(t) \parallel^2 \right)
	\label{equ:likelihood}
\end{align}
where $P=(N+1)^2$. Applying the logarithmic function  to 	\eqref{equ:likelihood},   the    log-likelihood function  will be obtained as:
   \begin{align}
   { L({\bm b};{\bf \Omega})} = -\frac{PN_s}{2}  \ln(\sigma^2) 
   -{1\over 2  \sigma^2} \sum_{t=1}^{N_s}\parallel{\bm  b}(t)-{{\bf Y}}({\bf \Psi })^T{\bm s}(t) \parallel^2 .
   \label{equ:LikelihoodFunc}
   \end{align}
   Therefore, the ML estimator of  ${\bf {\Omega}}$ can be formulated as:
   \begin{equation}
   {\bf \hat{\Omega}} =\mathop{\arg\max}\limits_{\bf \Omega}  \,	 { L({\bm b};{\bf \Omega})}.
   \label{equ:MLEst}
   \end{equation}
   In order to estimate  ${\bf \Psi}$,
%   and the dependency of ${ L({\bm b};{\bf \Omega})}$ to ${\bf \Psi}$ is through 
% $\parallel{\bm  b}(t)-{{\bf Y}^H}({\bf \Psi }){\bm s}(t) \parallel$ with only respect to ${\bf S}$, the
 %    to $\bf  \Theta$  and $\bf S$  is seen only at the last phrase of  (\ref{equ:LikelihoodFunc}) and ,
the     ML estimator in (\ref{equ:MLEst}) can be simplified as:
   \begin{align}
   \left({\bf \hat{\Psi},\hat{S}}\right)&=\mathop{\arg \min}\limits_{{\bf  \Psi,S}}  -L({\bm b};{\bf \Omega'}) \nonumber  \\&= \mathop{\arg \min}\limits_{{\bf  \Psi,S}} \sum_{t=1}^{N_s}\parallel{\bm b}(t)-{{\bf Y}}({\bf \Psi })^T{\bm s}(t) \parallel^2 ,
   \label{equ:OptUni}
   \end{align}
   where ${\bf \Omega'} \stackrel{\Delta}{=} \{{\bf \Psi},{\bf  S}\}$.
   Here ${\bm s}(t)$  and $\bf \Psi$  are the linear and non-linear parameters of our optimization problem,
%    due to the linear and non-linear relation of the observation vector${\bf b}_\mathrm{N}(t)$  with ${\bf s}(t)$  and $\bf \Theta$,
    respectively. Minimizing the objective function in     (\ref{equ:OptUni})     requires an exhaustive search in $(2L+LN_s)$-dimension space. 
    To decrease the computational complexity of such joint optimization problems, an iterative process is proposed based on  \cite{Chen_2008} as follows:
\\    1) Initialize $\bf \Psi$ and find the optimal estimator of ${\bm s}(t)$ as:
    \begin{equation}
    { \hat{\bm s}(}t) = {{{\bf Y}{\bf (\Psi)}}^T}^{\dag} {\bm b}(t)
    = \left({\bf Y}({\bf \Psi}) ^{T}{\bf Y}{(\bf \Psi)}\right)^{-1} {\bf Y}({\bf \Psi})^{T}{\bm b}(t).
    \label{equ:shatuni}
    \end{equation}
    where $\dag$  represents the pseudo inverse matrix.
   \\ 2) Estimate $\bf \Psi$ by putting the optimal estimation of $\bf \Psi$ in the  objective function as:
        \begin{equation}
    {\bf \hat{\Psi}}=\mathop{\arg \min}\limits_{{\bf  \Psi}}  \sum_{t=1}^{N_s}\parallel{\bm b}(t)- {{{\bf Y}{\bf (\Psi)}}^T}{{{\bf Y}{\bf (\Psi)}}^T}^{\dag}  {\bm b}(t) \parallel^2 .
    \label{equ:OptTheta}
    \end{equation}
    3) Estimate ${\bm s}(t)$ by considering the optimal estimation of $\bf \Psi$ obtained as \eqref{equ:shatuni}.
    \\
    4) Repeat  steps 2 and 3 until the difference of the objective function between two iterations  reaches a value lower than $T_{thr}$ limit.

    Algorithm \ref{Alg:Alg1} summarizes the  iterative
    ML estimator under uniform noise case.
    
    \begin{algorithm}
    	\caption{Iterative deterministic
    		ML estimator for Uniform Noise }
    	\label{Alg:Alg1}
    	\begin{algorithmic}[1]
    		\renewcommand{\algorithmicrequire}{\textbf{Input:}}
    		\REQUIRE ${\bm b}(t),1\leq t\leq N_s$, the $t$-th vector of observation.
    		\renewcommand{\algorithmicrequire}{\textbf{Output:}}
    		\REQUIRE $[{\bf \hat{\Psi}}]$, the vector of the estimated  DOAs.
    		\\ 
    		\STATE \textbf{Initialization:} Initialize
    		$[{\bf \hat{\Psi}}]^0$ 
    		randomly and $i=1$.
    		
    		\WHILE{$\Delta{ L}>T_{thr}$}
    		\STATE Obtain $	[{\bf \hat{S}}]^{i}$  using $	[{\bf \hat{\Psi}}]^{i-1}$.
    		\STATE Obtain $[{\bf \hat{\Psi}}]^{i}$ using $	[{\bf \hat{S}}]^{i}$, $[{\bf \hat{\Psi}}]^{i-1}$.
    		\STATE $[\hat{{\bf\Omega'}}]^{i} \leftarrow \{{[{\bf  \hat{\Psi}}]^i},{[{\bf \hat{S}}]^{i}}\}$
    		\STATE Compute ${ L({\bm b}};[\hat{{\bf\Omega}}']^{i})$ and then  $\Delta{ L} = { L({\bm b}};[\hat{{\bf\Omega}}']^{i}) - { L({\bm b}};[\hat{{\bf\Omega}}']^{i-1})$.
    		%		\IF {Merging leads to approching the target variance}
    		%		\STATE Merge two packets ($S'_{j}=S_{i}+S_{i+1}$);
    		%		\STATE Update mean and variance using \eqref{Eq:MeanUpdate} and \eqref{Eq:Update};
    		%		\STATE $ i\gets i+2 $ and $ j\gets j+1 $;
    		%		\ELSE
    		%		\STATE Keep the packet($S'_{j}=S_{i}$);
    		%		\STATE $ i\gets i+1 $ and $ j\gets j+1 $
    		%		\ENDIF
    		\ENDWHILE
    		\RETURN $[{\bf \hat{\Psi}}]$
    	\end{algorithmic}
    \end{algorithm}

   \subsection{Nonuniform Noise }
   \label{subsec:MLNonIden}
Now let the  noise vector   be zero mean Gaussian with covariance matrix of 
${\bf R}_{n}={\bf Q} = \mathrm{diag} \left\{q_1,q_2,\ldots,q_P\right\}$.
The  set of unknown parameters are defined as
${\bf \Omega} \stackrel{\Delta}{=} \{{\bf \Psi},{\bf  S}, {\bf Q}\}$.
% where ${\bf q} = [q_1,q_2,\ldots,q_P]^T$, $P=(N+1)^2$   and $q_p$ is the noise variance of the $p$'th channel of the HOA signal. Statistical distribution of noise vector is assumed to be zero mean Gaussian with covariance matrix ${\bf R}_\mathrm{n}= {\bf Q}$. Therefore, covariance matrix of HOA signal is:
%\begin{equation}
%{\bf Q} =\mathrm{diag} \left\{q_1,q_2,\ldots,q_P\right\}
%\end{equation}
The   likelihood function  will be:
\begin{align}
f({\bm b};{\bf \Omega}) =& \frac{1}{(2\pi)^{PN_s/2}  |\det({\bf Q})|^{N_s/2}}\nonumber\\\qquad\qquad &\qquad\times\exp\left(-{1\over 2} \sum_{t=1}^{N_s} {\bf g}(t)^T {\bf Q}^{-1} {\bf g}(t) \right), 
\label{equ:LikeliNonuni}
\end{align}
where ${ {\bm g}}(t) \stackrel{\Delta}{=}  { {\bm b}}(t) - {\bf {Y}}({\bf \Psi)}^T{\bm s}(t)$. Thus, The log-likelihood function will be:
\begin{equation}
{ L({\bm b};{\bf \Omega})} =-{PN_s \over 2}\ln(2\pi)-\frac{N_s}{2}\sum_{i=1}^{P}\ln(q_i) -{1\over 2  } \sum_{t=1}^{N_s}\parallel{ \tilde{\bm g}}(t)\parallel^2 ,
\label{equ:LogLikelihoodFunc}
\end{equation}
where 
%$\tilde{{\bm g}}(t) \stackrel{\Delta}{=} { \tilde{\bm b}}(t) - {\bf \tilde{Y}(\Psi)}{\bm s}(t)$.
% Logarithm of likelihood function will be:
%where we have:
\begin{align}
{ \tilde{\bm g}}(t) &\stackrel{\Delta}{=} {\bf Q}^{-1/2} {\bm {g}}(t) = { \tilde{\bm b}}(t) - {\bf \tilde{Y}}({\bf \Psi)}^T{\bm s}(t),  \\
{ \tilde{\bm b}}(t) &\stackrel{\Delta}{=} {\bf Q}^{-1/2} 	{ {\bm b}}(t),  \\
{ \tilde{\bf Y}(\bf \Psi)}^T &\stackrel{\Delta}{=}{\bf Q}^{-1/2}  {\bf Y}({\bf \Psi)}^T.
\end{align}
%Matrices ${\bf \tilde{g}}(t) $,${\bf \tilde{b}}_{N}(t)$  and ${\bf \tilde{Y}(\Theta)}$   can be known as whitened versions of ${\bf {g}}(t)$, ${\bf {b}}_{N}(t)$  and  $ {\bf Y(\Theta)}$.
 Therefore, the ML estimation of  $\bf \Omega$  can be written as:
\begin{equation}
{\bf \hat{\Omega}} =\mathop{\arg\max}\limits_{\bf \Omega}  	{ L(\bm b;\bf \Omega)}.
\label{equ:NonMLEst}
\end{equation}
To solve the optimization problem of \eqref{equ:NonMLEst}, an exhaustive search in $(2L+PN_s+P)$-dimensional  space is required. 
%Thus it is  highly complicated and impractical.
The iterative procedure mentioned in the previous subsection is employed to solve (\ref{equ:NonMLEst}).
% We exploit the iterative procedure mentioned in the previous subsection to solve 
%Since our goal is to find an estimation of $\bf \Theta$, we aim at decreasing the dimension of search space as much as possible, similar to the case of uniform noise.
% The variables ${\bf \Psi}$, ${\bf  S}$  and ${\bf Q}$ appear in the third term of the equation \eqref{equ:LogLikelihoodFunc}. 
 In order to perform ML estimation of  ${\bf \Psi}$, we need to consider ${\bf  S}$  and ${\bf Q}$ alongside altogether. 
 As a consequence, the sources' signals and the variances
%  both  source signals and noise variances
   must be estimated.
 Similar to the method introduced in       \ref{subsec:uniform}, first we fix $\bf \Psi$  and ${\bm s}(t)$, and then estimate the noise variances as a function of $\bf \Psi$  and  ${\bf s}(t)$. By replacing the estimated noise variances in the objective function, 
  the sources' signals is estimated and DOAs are obtained. This procedure is explained in the following.
%  we estimate the sources signals and finally obtain the DOAs. 
%  To find the derivative of (\ref{equ:LogLikelihoodFunc})
% respect to $q_p$,

Equation (\ref{equ:LogLikelihoodFunc}) can be simplified to 
%  Now, we rewrite (\ref{equ:LogLikelihoodFunc})  as follows:
 % firstly we rewrite
% ${\bf g}(t)=[{g}_1(t),\ldots,{g}_P(t)]^T$ using the definition (\ref{equ:LogLikelihoodFunc}):
\begin{equation}
{ L({\bm b};{\bf \Omega})} = -\frac{N_s}{2}\sum_{j=1}^{P}\ln(q_j) -{1\over 2  } \sum_{t=1}^{N_s} \sum_{j=1}^{P} \frac{(g_j(t))^2}{q_j},
\label{equ:LogLikelihoodFunc2}
\end{equation}
where ${\bm g}(t)=[{g}_1(t),\ldots,{g}_P(t)]^T$. The  derivative of  ${ L(\bm b;{\bf \Omega})} $ with respect to $q_p$ is calculated  as:
\begin{equation}
\frac{\partial  { L({\bm b};{\bf \Omega})} }{\partial q_p} = -{N_s \over 2}{1\over q_p}+{1\over 2}\sum_{t=1}^{N_s} \frac{(g_p(t))^2}{q_p^2}.
\label{equ:DerLogLik1}
\end{equation}
Letting (\ref{equ:DerLogLik1}) to be zero, the $p$'th noise variances can be found:
\begin{equation}
\hat{q}_p ={1\over N_s} \sum_{t=1}^{N_s} \left(g_p(t)\right)^2 
={1\over N_s} \|{\bm g}_p\|^2, \quad 1\leq p \leq P,
\label{equ:Qest}
\end{equation}
where ${\bm g}_p \stackrel{\Delta}{=} \left[g_p(1),g_p(2),\ldots,g_p(N_s)\right]^T.$  
Substituting $\hat{q}_p$   into (\ref{equ:LogLikelihoodFunc2}), ${ L(\bm b;{\bf \Omega})} $  is simplified to:
\begin{align}
{ L({\bm b};{\bf \Psi},\bm s}(t)) &= -\frac{N_s}{2}\sum_{j=1}^{P}\ln(\hat{q}_j) -{1\over 2  } \sum_{t=1}^{N_s} \sum_{j=1}^{P} \frac{(g_j(t))^2}{\hat{q}_j}\nonumber \\
%& =-\frac{N_s}{2}\sum_{j=1}^{P}\ln(\hat{q}_j) -{1\over 2  }  \sum_{j=1}^{P} {1\over \hat{q}_j}\sum_{t=1}^{N_s} (g_j(t))^2\nonumber \\
&=-\frac{N_s}{2}\sum_{j=1}^{P}\ln(\hat{q}_j) -{1\over 2  }  \sum_{j=1}^{P} {1\over \hat{q}_j}  \|{\bm g}_j\|^2 \nonumber \\
& =-\frac{N_s}{2}\sum_{j=1}^{P}\ln({1\over N_s} \|{\bm g}_j\|^2) -{1\over 2  } N_sP.
\end{align}
%\begin{align}
%%& =-\frac{N_s}{2}\sum_{j=1}^{P}\ln(\hat{q}_j) -{1\over 2  } N_sP \nonumber \\
%& =-\frac{N_s}{2}\sum_{j=1}^{P}\ln({1\over N_s} \|{\bm g}_j\|^2) -{1\over 2  } N_sP
%\end{align}
Therefore,  the  ML estimator of $\bf  \Psi$ and ${\bm s}(t)$ is given as:
\begin{equation}
\left({\bf \hat{ \Psi}},{\bm \hat{s}}(t)\right) = \mathop{\arg \min }\limits_{{\bf  \Psi},{\bm s}(t)} \sum_{j=1}^{P}\ln( \|{\bm g}_j\|^2)
\label{equ:OptNonuni}
\end{equation}
Similar to the optimization problem of (\ref{equ:OptUni}), ${ \hat{\bm s}}(t)$ can be presented as:
%This optimization problem is similar to  (\ref{equ:OptUni}).
%Thus we have
\begin{equation}
{ \hat{\bm s}}(t) ={\bf Y}{\bf (\Psi)}^{\dag} {\bm b}(t).
\label{equ:EstS}
\end{equation}
Substituting ${ \hat{\bm s}}(t)$  in (\ref{equ:OptNonuni}), the ML estimator of $\bf \Psi$ is obtained as:
\begin{equation}
{\bf  \hat{\Psi}} = \mathop{\arg \min }\limits_{{\bf  \Psi}} \sum_{j=1}^{P}\ln \left(  \|{ \hat{\bm g}}_j\|^2 \right)
\label{equ:OptThet2}
\end{equation}
where ${ \hat{\bm g}}(t) = {\bm b}(t) - {\bf Y}{\bf (\Psi)}^T {{\bf Y}{\bf (\Psi)}^T}^{\dag} {\bm b}(t)$
and 
 ${ \hat{\bm g}}_j$ is defined similar to ${ {\bm g}}_j$.
%  and ${\bf \hat{g}}(t)$ equals to:
%\begin{equation}
%{\bf \hat{g}}(t) = {\bf b}_\mathrm{N}(t) - {\bf Y}{\bf (\Theta)}  {\bf Y}{\bf (\Theta)}^{\dag} {\bf b}_\mathrm{N}(t)
%\label{equ:Ghat}
%\end{equation}

%To solve the non-linear optimization s in  (\ref{equ:NonMLEst}) and  (\ref{equ:OptThet2}) , expectation maximization (EM) technique is applied \cite{Lu2011}.    
% % % % % % % % % % % % % % % % % % % % % % % % % % % % % % % %
\subsection{Expectation Maximization Algorithm  for deterministic ML DOA Estimation for spatially  Nonuniform Noise }
In this subsection,
a new robust method based on EM algorithm is proposed for deterministic ML DOA estimation.
% we propose a new robust algorithm based on EM algorithm for deterministic ML DOA estimation. 
The EM algorithm is an iterative method for obtaining ML estimation, where the data model includes both observed and unobserved latent variables. 
First, this approach is examined for a single source case and then it will be extended  for multiple sources case. 
\subsubsection{Single Source Case}
%EM is a well-known algorithm to solve ML estimation. The optimization problem at hand can be solved iteratively using EM algorithm.
 As it can be seen in (\ref{equ:DataModel1}),
%  in the  signal model 
 the relationship between the received signal vector (incomplete data) 
${\bm b}$   
and the complete data 
${\bm b}^{(l)}\text{{  for }} 1 \leq l\leq L$ 
  will be:
 \begin{equation}
 {\bm b}  = \sum_{l=1}^{L}{\bm b}^{(l)},
 \end{equation}
 where 
  ${\bm b}^{(l)}$ 
   is the HOA signal received from  $l$'th source when only this source exists in the environment. According to (\ref{equ:DataModel1}), the     received signal model in the SH domain can be stated as:
%   the     received signal model in the SH domain can be obtained:
  \begin{equation}
  {{\bm b}}^{(l)}(t) = {{\bf y}}({\bf \Psi }_l)^T {s_{l}}(t) + {{\bm z}^{(l)}}(t),
  \end{equation}
%  for $t=1,\ldots,N_s$,
  where  ${{\bm z}^{(l)}}(t)$  is the  Gaussian noise vector in the sole presence of the $l$'th  source. 
%  an additive uncorrelated Gaussian noise vector at the presence of the $l$'th source, only.
    Considering (\ref{equ:NonMLEst}), the ML estimation will be
  \begin{equation}
  {\bf \hat{\Omega}}^{(l)}= \mathop{\arg \max}\limits_{{\bf {\Omega}}^{(l)}} { L}(\bm b^{(l)},{\bf {\Omega}}^{(l)}),
  \end{equation}
%  $  ({\bf \hat{\Theta}}^{(l)},{\bf \hat{s}}^{(l)}(t),{\bf \hat{q}}^{(l)}) $
  where 
  ${\bf \Omega}^{(l)} \stackrel{\Delta}{=} \{{\bf \Psi}_l,{\bf  S}^{(l)},  {\bf Q}^{(l)}\}$ and $L(\cdot)$ is defined in \eqref{equ:LogLikelihoodFunc}. 
  $   {\bf S}^{(l)} \stackrel{\Delta}{=} \{{  s_l(}1),\ldots,{  s_l(}N_s)\}$ and  
  ${\bf R}_\mathrm{n}^{(l)}={\bf Q}^{(l)} = \mathrm{diag} \{q^{(l)}_1,q^{(l)}_2,\ldots,q^{(l)}_P\}$
  is $l$'th sound source signal and the covariance matrix of the noise vector, respectively.  
%   ${\bf q}^{(l)} = [q^{(l)}_1,q^{(l)}_2,\ldots,q^{(l)}_P]^T$  is the noise variance vector, 
%   ${\bf {\Theta}}^{(l)} = [\theta_{l},\phi_{l}]^T$ and
%   ${\bf Q}^{(l)}\stackrel{\Delta}{=}\mathrm{E}\left\{{{\bf n}_{\mathrm{N}}^{(l)}}(t) \left({{\bf n}_{\mathrm{N}}^{(l)}}(t)\right)^{T} \right\}$.
 Similar to  \eqref{equ:Qest}, ${q}_p^{(l)}$ can be estimated as:
 \begin{equation}
 \hat{q}_p^{(l)} ={1\over N_s} \sum_{t=1}^{N_s} \left(g_p^{(l)}(t)\right)^2 ={1\over N_s} \|{\bm g}^{(l)}_p\|^2 ,\quad 1\leq  p \leq P,
 \label{equ:Qestl}
 \end{equation}
 where $ {\bm g}^{(l)}_p \stackrel{\Delta}{=} \left[g_p^{(l)}(1),\ldots,g_p^{(l)}(N_s)\right]^T$ and ${ {\bm g}}^{(l)}(t) \stackrel{\Delta}{=}  { {\bm b}}^{(l)}(t) - {\bf {y}}({\bf \Psi}_l)^T{ s_l}(t)$. The deterministic  ML estimation of single source DOA will be:
 \begin{equation}
 {\bf  \hat{\Psi}}_{l} = \mathop{\arg \min }\limits_{{\bf  {\Psi}}_{l}} \sum_{j=1}^{P}\ln\left( \|{ \hat{\bm g}}^{(l)}_j\|^2\right),
 \label{equ:OptThet3}
 \end{equation}
 where 
 $
 { \hat{\bm g}}^{(l)}(t) = {\bm b}^{(l)}(t) - {\bf y}{ (\bm\Psi_{l}})^T  {{\bf y}{(\bm\Psi_{l})}^T}^{\dag} {\bm b}^{(l)}(t).
 \label{equ:Ghatl}
 $
 
% we have:
% \begin{equation}
% {\bf g}^{(l)}_p = \left[g_p^{(l)}(1),g_p^{(l)}(2),\ldots,g_p^{(l)}(N_s)\right]^T.
% \label{equ:Vecgl}
% \end{equation}
 
%   \begin{align}
%   {\bf \tilde{b}}_\mathrm{N}^{(l)}(t)& = \left({\bf Q}^{(l)}\right)^{-1/2} 	{\bf {b}}_\mathrm{N}^{(l)}(t),  \label{equ:ObsHat}  \\
%   {\bf \tilde{Y}}\left({\bf \Theta}^{(l)}\right) &= \left({\bf Q}^{(l)}\right)^{-1/2}  {\bf Y}\left({\bf \Theta}^{(l)}\right)
%   \end{align}

      %%%%%%%%%%%%%%%%%%%%%%%%%%%%%%
      \subsubsection{Multiple Sources Case}
The EM algorithm for deterministic ML DOA estimation is expanded for multiple sources in this part.	Step by step procedure of the algorithm is explained as follows:
	\\
    \textbf{  Initialization}: Randomly initialize the direction of sources  $[{\bf  \hat{\Psi}}]^0$. The  matrices  $[{\bf \hat{Q}}]^0$  and $[{\bf \hat{Q}}^{(l)}]^0 $ are initialized as follows:
     \begin{equation}
     [{\bf \hat{Q}}^{(l)}]^0  = {1\over P} {\bf I}_P
%     \times \left[1,1,\ldots,1\right]^T \in \mathbb{R}^{P\times 1}
   \text{ and }
     [{\bf \hat{Q}}]^0 = {\bf I}_P.
%     \left[1,1,\ldots,1\right]^T \in \mathbb{R}^{P\times 1} 
     \label{equ:InitNoise}
     \end{equation}
     \\
     \textbf{Input} to the $i$'th loop: 
      $[{\bf \hat{Q}}^{(l)}]^{i-1}$  and $[{\bf  \hat{\Psi}}]^{i-1}$ .
      \\
      \textbf{Output} of the $i$'th loop: 
       $[{\bf \hat{Q}}^{(l)}]^{i}$  and $[{\bf  \hat{\Psi}}]^{i}$ .
       \\
       \textbf{Expectation step}:
%        Using the sources direction, matrix  
%       ${\bf Y}\left(\left[{\bf \hat{\Theta}}\right]^{i-1}\right)$
%         is formed and 
        The   noise covariance matrix  is obtained from the single source case ones as:
        \begin{equation}
        [{\bf \hat{Q}}]^{i-1} = \sum_{l=1}^{L} [{\bf \hat{Q}}^{(l)}]^{i-1}.
        \end{equation}
        The noise factor of the $l$'th single source, $\gamma^{(l)}$, is 
%          $\gamma^{(l)}$ that is proportional to the variance of noise added to each source
             calculated as:
         \begin{equation}
         \gamma^{(l)} = \frac{\mathrm{trace}\left([{\bf \hat{Q}}^{(l)}]^{i-1} \right)}{\mathrm{trace}\left(	[{\bf \hat{Q}}]^{i-1}\right)}.
         \end{equation}
        The HOA signal of each  source can be estimated as:
         \begin{align}
         { \hat{\bm b}}^{(l)}(t) &= \mathrm{E}\left\{{\bm {b}}^{(l)}(t)  \vert  {\bm {b}}(t)  \right\} 
         = \mathrm{E}\left\{{{\bf Y}}({\bf \Psi }_l)^T {s_{l}}(t) + {{\bm z}^{(l)}}(t) \vert  {\bm {b}}(t)  \right\} 
          \nonumber \\ 
         & \approx{{\bf Y}}([{\bf  \hat{\Psi}}_l]^{i-1})^T\hat{s}_{l}(t)+ \gamma^{(l)} \left({\bm {b}}(t)  -  {\bf Y} ([{\bf  \hat{\Psi}}]^{i-1})^T { \hat{\bm s}}( t)\right),
         \label{equ:bhatlt}
         \end{align}
         where ${ \hat{\bm s}}(t)$   is obtained using (\ref{equ:EstS}).
         %and then  $\hat{s}_{l}(t)$ .
         \\
         \textbf{Maximization step}: The goal of this step is to find $[{\bf  \hat{\Psi}}]^{i}$. 
%         Based on equations (\ref{equ:SpheHarmMat}) and (\ref{equ:SpheHarmVec}),
        The  vector  ${\bf Y} ([{\bf  \hat{\Psi}}_l]^{i})$ can be obtained as a function  of  $[{\bf  \hat{\Psi}}_l]^{i}$.
%                  then based on it, below vectors are formed:
%          \begin{align}
%          {\tilde{{\bf Y}}}\left(\left[{\bf \hat{\Theta}}^{(l)}\right]^{i}\right) &\stackrel{\Delta}{=} 	\left({\bf \hat{Q}}^{(l)}\right)^{(1/2)} {{{\bf Y}}}\left(\left[{\bf \hat{\Theta}}^{(l)}\right]^{i}\right),\quad 1\leq l \leq L \\
%          {\bf \tilde{\hat{b}}}_\mathrm{N}^{(l)}(t) &= \left({\bf \hat{Q}}^{(l)}\right)^{(1/2)} {\bf {\hat{b}}}_\mathrm{N}^{(l)}(t) ,\quad 1\leq l \leq L ,
%          \end{align}
          Then, $ { \hat{\bm g}}^{(l)}(t)$ is constructed: 
           \begin{align}
           { \hat{\bm g}}^{(l)}(t) &=  { {\hat{\bm b}}}^{(l)}(t) - 	{{{\bf Y}}}\left([{\bf  \hat{\Psi}}_l]^{i}\right)^T	{{{\bf Y}}}{\left([{\bf  \hat{\Psi}}_l]^{i}\right)^T} ^{\dag} { {\hat{\bm b}}}^{(l)}(t)  \label{equ:GhatlAlg}  \text{ \,and}\\
           { \hat{\bm g}}^{(l)}_j &= \left[\hat{g}_j^{(l)}(1),\hat{g}_j^{(l)}(2),\ldots,\hat{g}_j^{(l)}(N_s)\right]^T.
           \label{equ:VecglAlg}
           \end{align}
%            \begin{align}
%            {\bf \hat{g}}^{(l)}(t) &=  {\bf \tilde{\hat{b}}}_\mathrm{N}^{(l)}(t) - 	{\tilde{{\bf Y}}}\left(\left[{\bf \hat{\Theta}}^{(l)}\right]^{i}\right)	{\tilde{{\bf Y}}}\left(\left[{\bf \hat{\Theta}}^{(l)}\right]^{i}\right) ^{\dag} {\bf \tilde{\hat{b}}}_\mathrm{N}^{(l)}(t)  \label{equ:GhatlAlg} \\
%            {\bf  \hat{g}}^{(l)}_i &= \left[\hat{g}_i^{(l)}(1),\hat{g}_i^{(l)}(2),\ldots,\hat{g}_i^{(l)}(N_s)\right]^T.
%            \label{equ:VecglAlg}
%            \end{align}
           Note that $ {  \hat{\bm g}}^{(l)}_i $  is  a function of  $[{\bf  \hat{\Psi}}_l]^{i}$.  Therefore, the optimization problem to find $[{\bf  \hat{\Psi}}_l]^{i}$  will be:
           \begin{equation}
           [{\bf  \hat{\Psi}}_l]^{i}= \mathop{\arg \min }\limits_{[{\bf  {\Psi}}_l]^{i}} \sum_{j=1}^{P}\ln\left( \|{\bf \hat{g}}^{(l)}_j\|^2\right),\quad 1\leq l \leq L.
           \label{equ:OptThetAlg}
           \end{equation}
           After Estimating $[{\bf  \hat{\Psi}}_l]^{i}$,  the vector  ${\bf Y} ([{\bf  \hat{\Psi}}_l]^{i})$  can be obtained. According  to (\ref{equ:Qestl}), the elements of the noise variance vector are estimated as:
            \begin{equation}
            \left[\hat{q}_p^{(l)}\right]^i = {1\over N_s} \left\|\hat{{\bm g}}^{(l)}_p\right\|^2,\quad 1\leq  p \leq P,
            \label{equ:QestlAlg}
            \end{equation}
%            and obtain 
            \begin{equation}
        [{\bf \hat{Q}}^{(l)}]^{i} = \mathrm{diag} \left\{[q^{(l)}_1]^{i} ,\ldots,[q^{(l)}_P]^{i}\right\}.
            \end{equation}
%            Now we can say:
%             \begin{equation}
%             \left[{\bf \hat{q}}^{(l)}\right]^{i} = \left[	\left[\hat{q}_1^{(l)}\right]^i ,	\left[\hat{q}_2^{(l)}\right]^i ,\ldots,	\left[\hat{q}_P^{(l)}\right]^i \right]^T
%             \end{equation}

             Using the EM  algorithm, the ML estimation of DOA  of each source can be estimated separately. 
%             Notice that using this solution, the optimization problem introduced in (\ref{equ:OptThet2}) that requires an exhaustive search in the  $2L$-dimension space is simplified to $L$ 2-dimensional optimization problems. 
             By comparing  (\ref{equ:OptThetAlg}) with  (\ref{equ:OptThet2}), it can be seen that the search space is reduced from $2L$-dimensional  in (\ref{equ:OptThet2})  to $L$ 2-dimensional in (\ref{equ:OptThetAlg}).
             This improvement significantly decreases the optimization complexity. 
             
             The proposed EM algorithm for deterministic ML DOA
             estimation for nonuniform noise case is summarized in Algorithm \ref{Alg:Alg2}.
             
%             Algorithm \ref{Alg:Alg2} summarizes the  EM algorithm for DOA
%             ML estimation  under nonuniform noise assumption.
             
              \begin{algorithm}
             	\caption{
             	EM algorithm for deterministic
             	ML estimation  for  nonuniform noise }
             	\label{Alg:Alg2}
             	\begin{algorithmic}[1]
             		\renewcommand{\algorithmicrequire}{\textbf{Input:}}
             		\REQUIRE ${\bm b}(t),1\leq t\leq N_s$, the $t$-th vector of observation.
             		\renewcommand{\algorithmicrequire}{\textbf{Output:}}
             		\REQUIRE $[{\bf \hat{\Psi}}]$, the vector of the estimated  DOAs.
             		\\ 
             		\STATE \textbf{Initialization:} Initialize
             		$[{\bf \hat{\Psi}}]^0$ 
             		randomly, $i=1$,     $[{\bf \hat{Q}}]^0$  and $[{\bf \hat{Q}}^{(l)}]^0 $  as follows:
             		\begin{equation}
             		[{\bf \hat{Q}}^{(l)}]^0  = {1\over P} {\bf I}_P
             		\text{ and }
             		[{\bf \hat{Q}}]^0 = {\bf I}_P.
             		\label{equ:InitNoise}
             		\end{equation}
             		
             		\WHILE{$\Delta{ L}>T_{thr}$}
             		\STATE \textbf{Expectation step}:
             		\vspace{3pt}
        \STATE $ [{\bf \hat{Q}}]^{i-1} \leftarrow \sum_{l=1}^{L} [{\bf \hat{Q}}^{(l)}]^{i-1}$.
        \vspace{3pt}
        \STATE $\gamma^{(l)} \leftarrow \frac{\mathrm{trace}\left([{\bf \hat{Q}}^{(l)}]^{i-1} \right)}{\mathrm{trace}\left(	[{\bf \hat{Q}}]^{i-1}\right)}$.
             		\STATE Obtain $	{ \hat{\bm b}}^{(l)}(t)$  using \eqref{equ:bhatlt}.
             		\STATE Obtain $	{ \hat{\bm s}}(t)$  using \eqref{equ:EstS}.
             		\STATE \textbf{Maximization step}:
             		\STATE ${ \hat{\bm g}}^{(l)}(t) \leftarrow  { {\hat{\bm b}}}^{(l)}(t) - 	{{{\bf Y}}}\left([{\bf  \hat{\Psi}}_l]^{i}\right)^T	{{{{\bf Y}}}\left([{\bf  \hat{\Psi}}_l]^{i}\right) ^T}^{\dag} { {\hat{\bm b}}}^{(l)}(t)$
             		\STATE $[{\bf  \hat{\Psi}}_l]^{i} \leftarrow \mathop{\arg \min }\limits_{[{\bf  {\Psi}}_l]^{i}} \sum_{j=1}^{P}\ln\left( \|{\bf \hat{g}}^{(l)}_j\|^2\right),\quad 1\leq l \leq L$
             		\STATE  $\left[\hat{q}_p^{(l)}\right]^i \leftarrow {1\over N_s} \left\|\hat{{\bm g}}^{(l)}_p\right\|^2,\quad 1\leq  p \leq P$
             		\STATE $ [{\bf \hat{Q}}^{(l)}]^{i} \leftarrow \mathrm{diag} \left\{[q^{(l)}_1]^{i} ,\ldots,[q^{(l)}_P]^{i}\right\}$
             		\vspace{3pt}
             		\STATE Compute ${ L({\bm b}};[\hat{{\bf\Omega}}]^{i})$ and then  $\Delta{ L} = { L({\bm b}};[\hat{{\bf\Omega}}]^{i}) - { L({\bm b}};[\hat{{\bf\Omega}}]^{i-1})$.
             		%		\IF {Merging leads to approching the target variance}
             		%		\STATE Merge two packets ($S'_{j}=S_{i}+S_{i+1}$);
             		%		\STATE Update mean and variance using \eqref{Eq:MeanUpdate} and \eqref{Eq:Update};
             		%		\STATE $ i\gets i+2 $ and $ j\gets j+1 $;
             		%		\ELSE
             		%		\STATE Keep the packet($S'_{j}=S_{i}$);
             		%		\STATE $ i\gets i+1 $ and $ j\gets j+1 $
             		%		\ENDIF
             		\ENDWHILE
             		\RETURN $[{\bf \hat{\Psi}}]$
             	\end{algorithmic}
             \end{algorithm}

             %%%%%%%%%%%%%%%%%%%%%%%%%%%%
             \section{ Cramer-Rao Bound}
%              Cramer-Rao bound (CRB) is one of the most important tools for performance assessment of unbiased estimators.
               In this section,   CRB  of the  deterministic 
               DOA estimator will be derived for a signal model with the spatially nonuniform noise.
               %               DOA estimator 
%             for the spatially  nonuniform noise   will be derived. 
             %              is  Note that the Gaussian noise in the received signal model in HOA domain follows a non-uniform model.
               This  work is the extension  of \cite{Kumar_2015} and  \cite{Chen_2008} to the  deterministic model of the sound sources in the SH domain.
%               for statistical model of sound sources in HOA domain. 

             \begin{theorem}
             	\label{theo:CRLB}
             	 CRB  of the  deterministic ML DOA estimator
             	 for spatially nonuniform noise in the SH domain  is given by
%             Cramer-Rao bound for estimation of each one of sound source directions in HOA domain is:
             \begin{align}
             \mathrm{var}(\theta_l) &\geq  [C_1]_{ll}, l=1,\cdots,L,\\
             \mathrm{var}(\phi_l) &\geq  [C_2]_{ll}, l=1,\cdots,L,
             \end{align}
            in  which 
             \begin{align}
             C_1 &= ({\bf F}_{{\bm \theta,\bm\theta}}-{\bf F}_{{\bm \theta,\bm\phi}} {\bf F}_{{\bm \phi,\bm\phi}}^{-1} {\bf F}_{{\bm \phi,\bm\theta}})^{-1},  \\
             C_2 &= ({\bf F}_{{\bm \phi,\bm\phi}}-{\bf F}_{{\bm \phi,\bm\theta}} {\bf F}_{{\bm \theta,\bm\theta}}^{-1} {\bf F}_{{\bm \theta,\bm\phi}})^{-1} ,
             \end{align}
             \begin{equation}
             {\bf F}_{\bm{\alpha},\bm{\beta}} =  {\bf S}_s   \odot \left(  {\dot {\bf Y}}_{{\bm \alpha}} {\bf {C}}_{b}^{- 1}    {\dot {\bf Y}}_{{\bm \beta}}^T  \right),
             \end{equation}
             where $ {\bm \alpha}$ and  ${\bm \beta}$ can 
             be equal to ${\bm \theta}$ or ${\bm \phi}$  independently,
%             take independently ${\bm \theta}$ or ${\bm \phi}$
              ${\bf {C}}_{b}$ is the covariance of  the HOA signal,  $\bm S_s= \sum_{t=1}^{N_s}   {\bm s}(t)  {\bm s}(t)^T  $ and
             \begin{equation}
             [\dot{{\bf Y}}_{\bm \alpha}]_{ij}  = {\partial \over \partial \alpha_j}[{\bf Y}]_{ij},i=1,\ldots,P \, \mathrm{and}\, j=1,\ldots,L.
             \end{equation}
              \end{theorem}
\begin{proof} 
	See Appendix.
\end{proof}              
%	\ref{Append:TheoremProof}
		
             \section{Simulation}
             In this section, the proposed EM algorithm is evaluated and compared with the traditional standard narrow-band MUSIC algorithm
             \cite{Schmidt_1986} and the recently proposed ICA based method \cite{Epain_2012} through various            scenarios.                              
%             Now, we evaluate the proposed EM  algorithm  and compare
%             it with the traditional   standard narrow-band MUSIC  algorithm             \cite{Schmidt_1986} and  the recent ICA based  method \cite{Epain_2012}  through various scenarios.
% َAlso, we declared that the RMSE of estimator under uniform noise variaces assumption be larger than RMSE of estimator under nonuniform case. 
%             
             In the conducted simulations, the  SMA  is an open array  of radius 15 cm  consisting   12 omnidirectional microphones.    The SMA is located at the coordinates (5 m, 7 m, 1.5 m) of a room of size 8 m $\times$ 10 m $\times$ 3 m. Two sound sources are located at 2 m distance of the array at the angular locations of  
             $(\phi^\circ,\theta^\circ)  =  (40,70)$ and $(70,60)$. 
             Configuration of the room and the position of the SMA and sources in the simulation setup are demonstrated in Fig. 	\ref{fig:simulmodel}.
             The signal to reverberation ratio (SRR) is
             almost equal to -3.5 dB and the room reverberation time (RT60) is approximately 400 ms.
             The sources play the speech signals with duration about 1 s which are sampled at 16 KHz. The impulse response for the room between the sources and the SMA is calculated using {MCRoomSim}, which is a  multichannel room acoustics simulator  \cite{wabnitz2010room}. Microphone signals and additive white Gaussian noise  are filtered with the HOA encoding filters which result in  2nd order HOA signals and SH domain noise, respectively. The length of the  HOA encoding filters is 512 and designed such that its output SNR is maximized. Then, the HOA signals
             are filtered by bandpass filters with the pass-band of 500 to 3500 Hz.  The optimization in     \eqref{equ:OptTheta} and            \eqref{equ:OptThetAlg} are performed by Nelder-Mead direct search method \cite{Lagarias_1998}.  The 
             FastlCA is used for applying the ICA algorithm in the MATLAB environment \cite{Fastica}.

\begin{figure}
	\centering
	\includegraphics[width=0.85\linewidth]{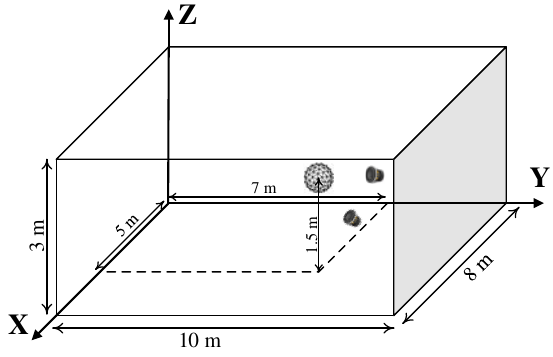}
	\caption{Configuration of the room and the position of the SMA and sources in our simulation}
	\label{fig:simulmodel}
\end{figure}
% \begin{table}[H]
%          	\centering
%          	\caption{Sound sources direction }
%          	\begin{tabular}{c | c c c  c}
%          		\toprule
%          		Source number  &   1 & 2 & 3 & 4 \\ \hline 
%          		$(\phi^\circ,\theta^\circ)$ & (40,90) &  (70,-90) &(60,80) & (90,-90) \\
%          		\bottomrule     		
%          	\end{tabular}
%          	\label{tab:SouDirDEM}
% \end{table}
 
 In  Figs. 	\ref{fig:snapshotrmsethet} and 
 \ref{fig:snapshotrmsephi}, the average root mean square error (RMSE)  of the  estimating $\bm \theta$
   and  $\bm \phi$
 for 50 different realizations in 30 dB SNR for EM, ICA, MUSIC,  ML estimation for uniform noise case (see Alg. \ref{Alg:Alg1}) and CRB versus the number of snapshots $N_s$  are presented.
 The MUSIC algorithm does not converge below 5000 snapshots.
%  For lower 5000 snapshots the MUSIC algorithm does not converge. 
  As the number of snapshots increase, the RMSE of estimation decreases for all DOA estimation methods. As  expected, the EM algorithm outperforms the uniform noise case estimation, due to this fact that the nonuniform noise well matches the signal model in the SH domain.

 \begin{figure}
 	\centering
 	\includegraphics[width=1\linewidth]{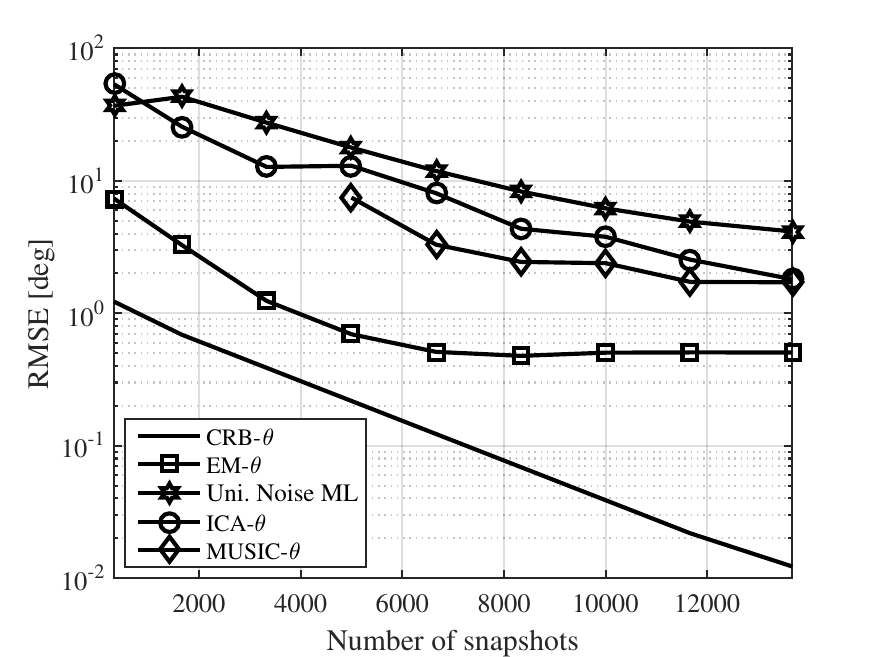}
 	\caption{RMSE for $\theta$ estimation versus the number of snapshots in 30 dB SNR}
 	\label{fig:snapshotrmsethet}
 \end{figure}
 
 \begin{figure}
 	\centering
 	\includegraphics[width=1\linewidth]{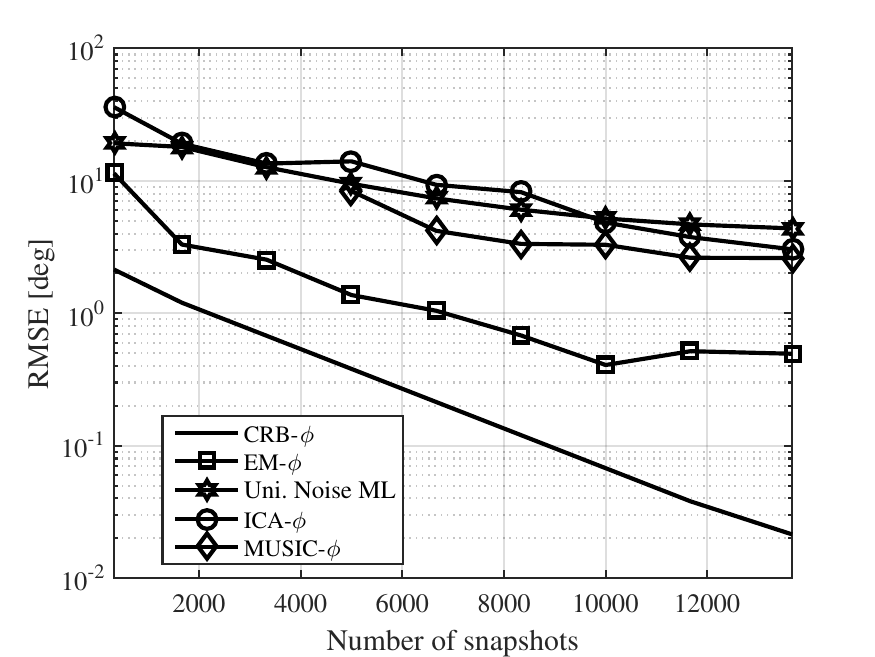}
 	\caption{RMSE for $\phi$ estimation versus the number of snapshots in 30 dB SNR}
 	\label{fig:snapshotrmsephi}
 \end{figure}

 The estimated $\bm \theta$
 and  $\bm \phi$ RMSE for the EM algorithm  as a function of the SNR through boxplot representation is plotted in Figs. \ref{fig:boxplthetaem} and  \ref{fig:boxplphiem}, respectively.
 % In Figs. \ref{fig:boxplthetaem} and  \ref{fig:boxplphiem},  we plot the RMSE of estimating
%  $\bm \theta$
% and  $\bm \phi$ for the EM algorithm  as a function of the SNR through boxplot representation.
  In these figures,   for  each SNR value,  the  RMSE is obtained with the average of 50 different realizations of the proposed algorithm.  
 The box has lines at  the lower, median and upper quartile values of the RMSE. The  whiskers are lines extending from each end of the box to show  the extent of the rest of the values. Outliers are the values outside
 the ends of the whiskers. If there is no value outside the whisker,   a dot is placed at the bottom whisker. As can be seen, SNR is increased by decreasing the RMSE variance.
% more increase in SNR, more decrease of  the  RMSE variance.
 %and it closes to its mean value. 
 %with the increase in  SNR  the range of RMSE changes was decreased. 
 Therefore in higher SNR, the estimated value is likely  closer to the true value.

\begin{figure}
\centering
\includegraphics[width=1\linewidth]{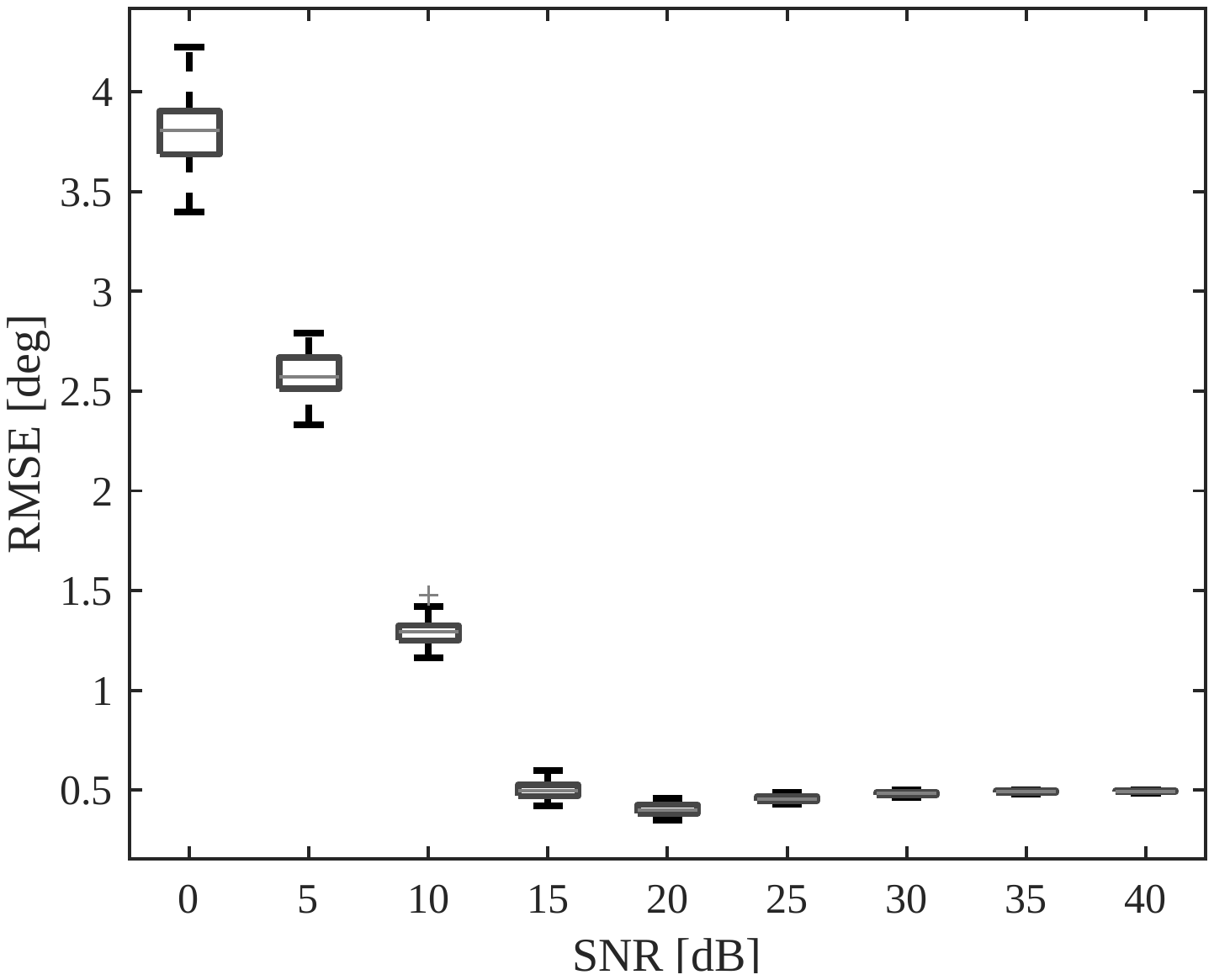}
\caption{RMSE Box-plot of EM  estimation of $\theta$}
\label{fig:boxplthetaem}
\end{figure}
   
\begin{figure}
\centering
\includegraphics[width=1\linewidth]{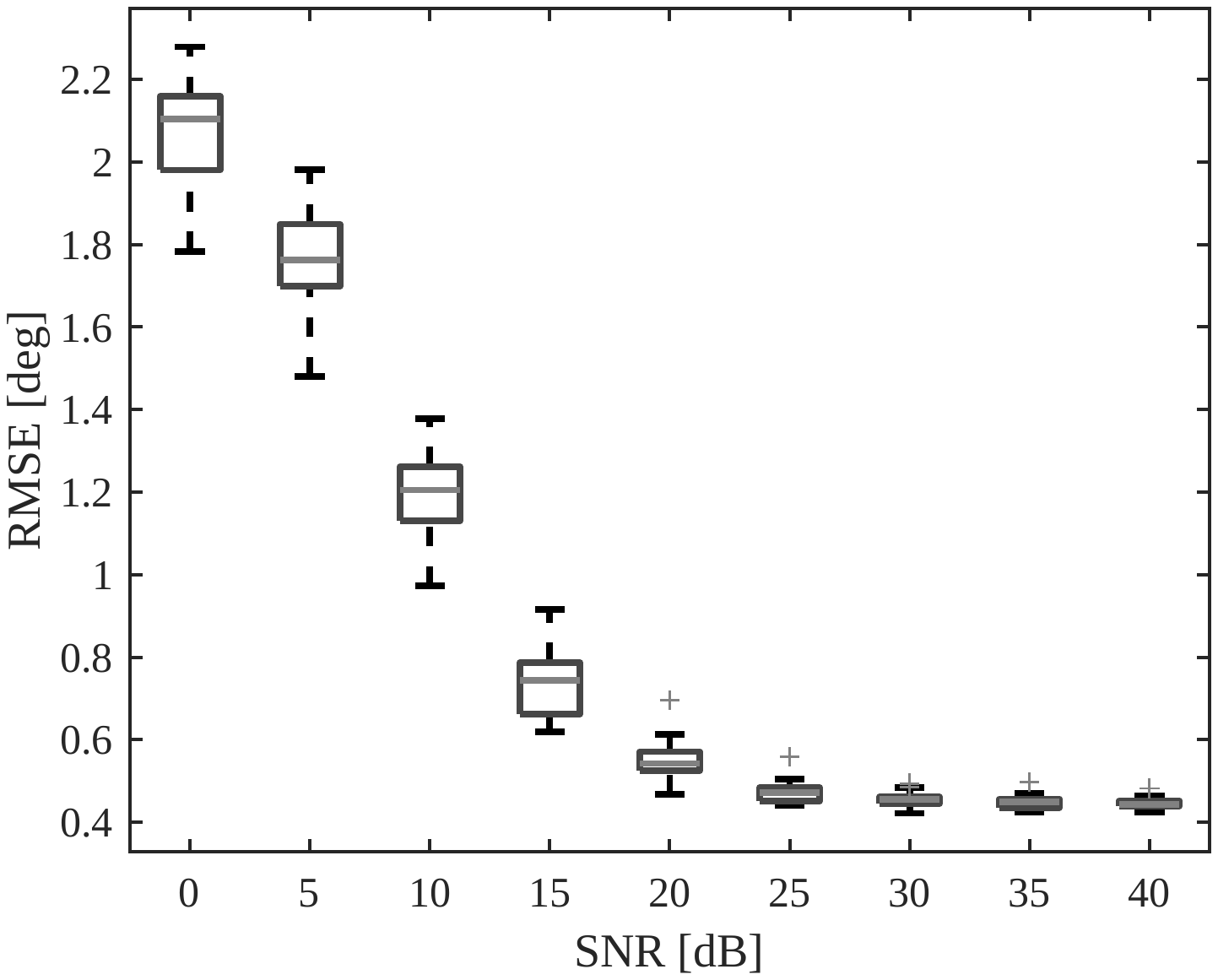}
\caption{RMSE Box-plot of EM  estimation of $\phi$}
\label{fig:boxplphiem}
\end{figure}

 The estimated $\bm \theta$  and  $\bm \phi$ RMSE for EM,
 ICA, MUSIC, uniform ML estimation and CRB versus SNR is shown in Figs.
 \ref{fig:totcompdeterthet} and  \ref{fig:totcompdeterphi}, respectively.
% 
% In Figs.
% \ref{fig:totcompdeterthet} and  \ref{fig:totcompdeterphi},
%   RMSE of  estimating
%  $\bm \theta$  and  $\bm \phi$  for EM, ICA, MUSIC,  ML estimation for uniform noise case and CRB versus SNR is shown. 
 The range of SNR values is between 0 to 40 dB.  The average of 50 different realizations
 is used to achieve each simulated point.
 As shown, the proposed algorithm 
 is  closer to the CRB compared to MUSIC and ICA.
 Performance of  the {ICA} method is highly dropped in low SNR values due to not considering the environmental noise.
% and its RMSE  lays above the MUSIC by a significant increase.
  In higher SNR values, the ICA assumption becomes closer to the reality, resulting in the ICA outperforms the MUSIC.  The  EM algorithm exhibits a better performance because of considering the environmental noise  and reverberation. The signal is assumed to be independent and non-Gaussian for the ICA algorithm. But due to the reverberation,  both assumptions are not realistic for DOA estimation. 
 In order to show that the distribution of the HOA signal is Gaussian,  Kolmogorov-Smirnov  hypothesis  test is used. The test result, with  the 5\% significance level, confirms that 
 the HOA signals come from a Gaussian distribution.
%   with the One-sample Kolmogorov-Smirnov (KS) hypothesis  test with  the 5\% significance level. 
% If KS test is passed  considering specific significance level, the signal comes from a standard normal distribution. The HOA signal passes the test and our claim is confirmed.
  Also, the cross-correlation coefficient between the first and the second normalized HOA signals are calculated  as 
 $\rho= \mathrm{E}\{b_1 b_2\}$ which is equal to $0.834$.  For better visualization, the 
 histogram of the second HOA signal  and 
 the cross-correlation between the first and second  HOA signals are plotted in
  Figs.
  \ref{fig:DenFunHOA2} and   \ref{fig:Xcorr12}, respectively.
 % in Figs.
%% \ref{fig:DenFunHOA2} and  
%  \ref{fig:Xcorr12} 
%%   histogram of the second HOA signal  and 
%   the cross correlation between the first and second  HOA signals   are demonstrated. According to this figure,  the HOA signals can not be disregarded
% 
% According to Figs.  \ref{fig:DenFunHOA2} and   \ref{fig:Xcorr12}, the distribution of HOA signal tends to Gaussian and the cross correlation between the HOA signals can not be disregarded. 
  Considering the   signal model matches to the HOA domain, the EM algorithm can achieve lower RMSE  for  estimating
 $\bm \theta$  and  $\bm \phi$  and    is closer to the CRB.  Referring to these results, we can say that the
 EM algorithm  shows at least an improvement of 6 dB in
 robustness compared to the best of MUSIC and ICA methods
 in the noisy environments.

\begin{figure}
\centering
\includegraphics[width=1\linewidth]{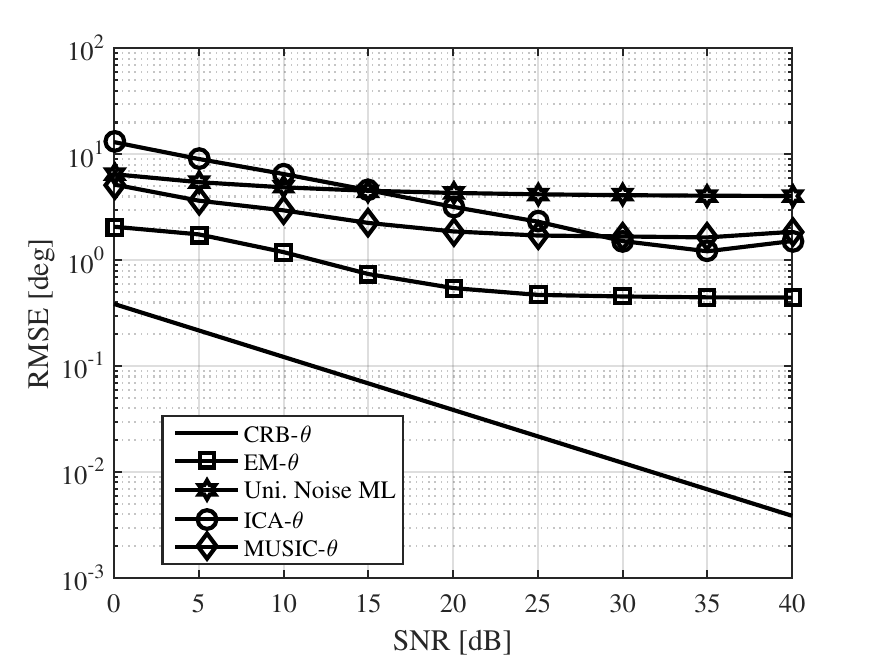}
\caption{
	RMSE comparison of $\theta$ estimation in  EM
	versus ICA and MUSIC methods along with the CRB.
	}
\label{fig:totcompdeterthet}
\end{figure}             
%\vspace{-cm}
\begin{figure}
\centering
\includegraphics[width=1\linewidth]{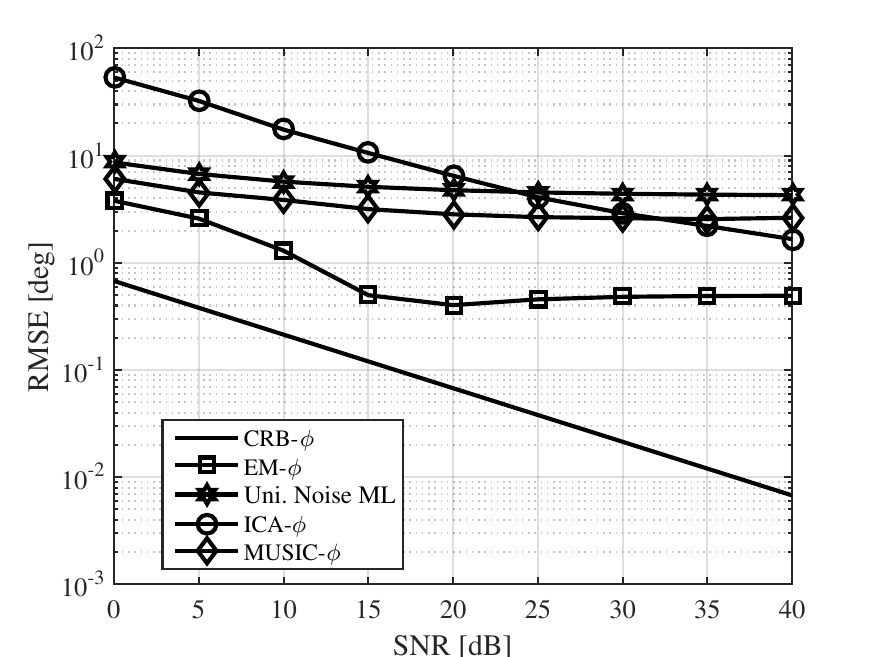}
\caption{
	RMSE comparison of $\phi$ estimation in  EM
	versus ICA and MUSIC methods along with the CRB.}
\label{fig:totcompdeterphi}
\end{figure}
%\begin{figure}
%	\centering
%	\includegraphics[width=1\linewidth]{Pic/DenFunHOA2}
%	\caption{
%		Histogram of the second   HOA signal}
%	\label{fig:DenFunHOA2}
%\end{figure}
%\begin{figure}
%	\centering
%	\includegraphics[width=1\linewidth]{Pic/Xcorr12}
%	\caption{
%		Cross correlation between the first and second  HOA signals}
%	\label{fig:Xcorr12}
%\end{figure}     
\begin{figure}
	\centering
	\begin{tabular}[c]{cc}
		\begin{subfigure}[b]{0.45\textwidth}
			\includegraphics[width=\textwidth,height=.15\textheight]{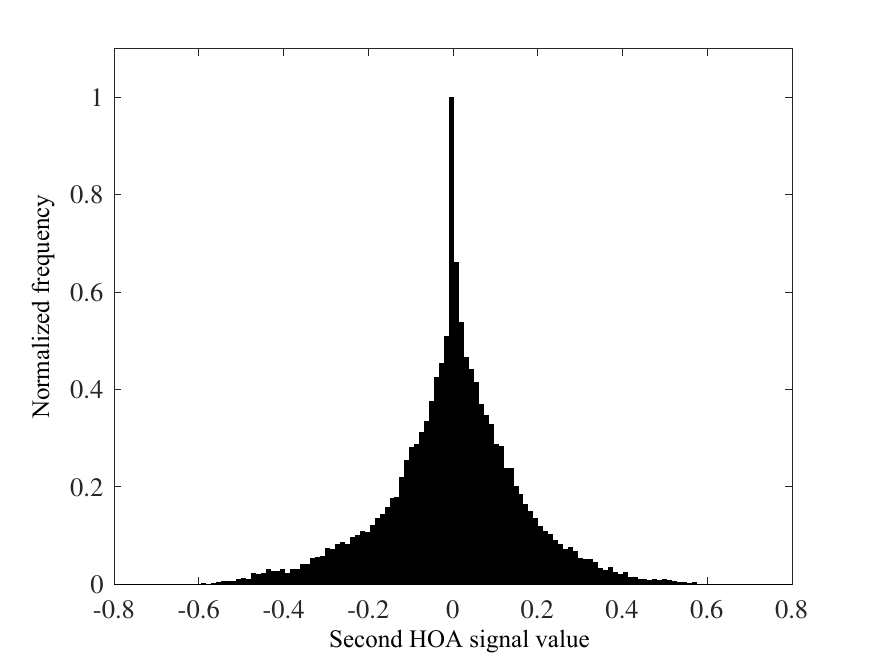}
			\caption{Histogram of the second   HOA signal}
			\label{fig:DenFunHOA2}
		\end{subfigure}\\
		\begin{subfigure}[b]{0.45\textwidth}
			\includegraphics[width=\textwidth,height=.15\textheight]{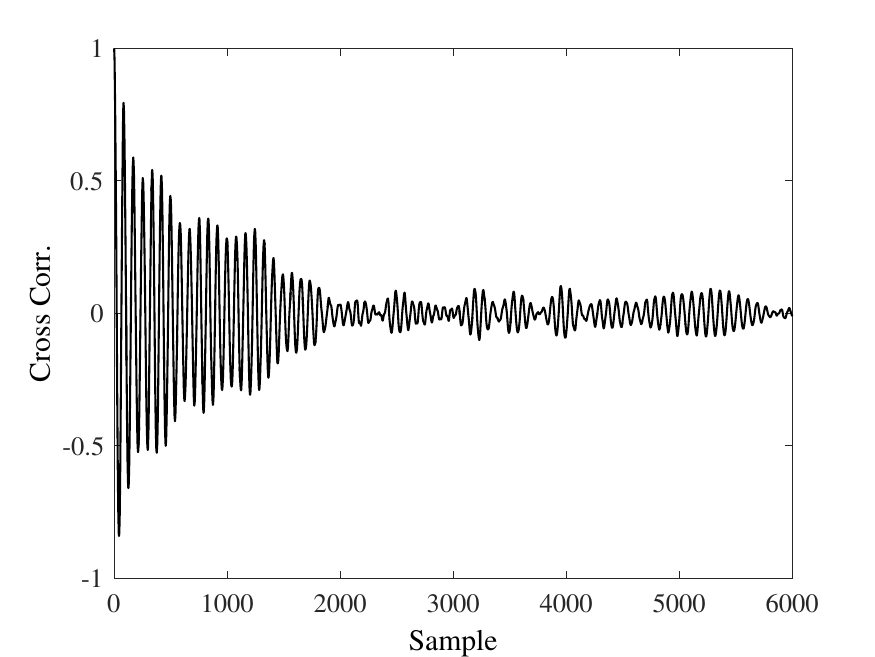}
			\caption{	Cross-correlation between the first and second  HOA signals}
			\label{fig:Xcorr12}
		\end{subfigure}
	\end{tabular}
	\caption{Evaluation of HOA signal}\label{fig:eval}
\end{figure}

 To examine the robustness of  the EM algorithm in  the reverberant environments, the average RMSE of the  estimating 
$\bm \theta$  and  $\bm \phi$  for EM, ICA and  MUSIC
 for 50 different realizations in 40 dB SNR 
 versus different RT60 are reported in  Table
    		\ref{tab:RMSRVBThet} and \ref{tab:RMSRVBPhi}.
 %As it can be seen,
 In the lower RT60s,  the  ICA and EM algorithm almost have the same performance. Because the ICA method does not consider
 the correlation in the HOA signals, its RMSE grows by increasing the RT60.
 %correlation increase in the received signal.
 Also, the MUSIC algorithm shows an acceptable  performance in the lower RT60 and degrades as RT60 increases.
 According  to Tables 	\ref{tab:RMSRVBPhi} and 	\ref{tab:RMSRVBThet},  the proposed algorithm 
 demonstrates  at least an improvement of  7 dB  in robustness  compared to the MUSIC and ICA methods in  the reverberant environments   as it was in  the noisy environments.

 \begin{table}
 	\normalsize
 	\centering
 	\caption{RMSE estimating  $\bm \theta$  (degree) in 40 dB SNR  versus different RT60s
 	}
 	\begin{tabular}{c |c| c| c}
 		\toprule
 		RT60
 		{[sec]}  &  {EM} & {ICA} & {MUSIC}  \\
 		\hline \hline
 		0.110 & 0.183 & 0.186 & 0.893 \\ 
 		0.241 & 0.249 & 0.502& 0.998 			\\  
 		0.301 &  0.290 & 0.870 & 1.564 \\ 
 		0.411 & 0.344 & 1.469 & 2.379\\
 		0.650 & 0.386 & 1.877 & 2.500 \\
 		0.799 & 0.410 & 2.198 &  2.834 \\
 		\bottomrule
 	\end{tabular}
 	\label{tab:RMSRVBThet}
 \end{table}
 
  \begin{table}
 	\normalsize
 	\centering
 	\caption{RMSE of estimating  $\bm \phi$  (degree) in 40 dB SNR versus different RT60s
 	}
 	\begin{tabular}{c| c| c| c}
 		\toprule
 		RT60
 		{[sec]}  &  {EM} & {ICA} & {MUSIC}  \\
 		\hline \hline
 		0.110 & 0.192 & 0.218 & 1.001 \\ 
 		0.241 & 0.263 & 0.701& 1.404  			\\  
 		0.301 &  0.291 & 1.000 & 1.500 \\ 
 		0.411 & 0.320 & 2.691 & 1.579 \\
 		0.650 & 0.396 & 2.740 & 1.681 \\
 		0.799 & 0.457 & 2.488 &  2.623 \\
 		\bottomrule
 	\end{tabular}
 	\label{tab:RMSRVBPhi}
 \end{table}

\section{Conclusion}
\label{sec:cncl}    
In this paper,   considering   the general model of the received signal in the SH domain, the EM algorithm is proposed for deterministic  ML estimation of DOA  of multiple sources in the presence of spatially nonuniform noise.  In order to reduce the complexity of the ML estimation, the algorithm is broken down into   two expectation and maximization steps.
In the expectation step,  the HOA signal of each single source case (latent variable) is obtained from the observed HOA signal. In the maximization step,  the DOA of each source is estimated using the corresponding  HOA signal which is obtained in the expectation step. 
The simulation results demonstrated that the proposed algorithm shows at least 6  and 7 dB better  robustness  in terms of RMSE  in the  reverberant and noisy environments, respectively, compared to the  MUSIC and ICA methods. 
Estimation of DOA through machine
learning algorithms in the HOA domain is a part of our future work.
%As the future work, we are going to  estimate DOA through machine learning algorithms in the HOA domain.
%analyze the convergence of the proposed method.

\appendix[Proof of   Theorem \ref{theo:CRLB}]
  	First,  we define the unknown parameters  vector as 
  \begin{align}
  {\bf\Theta} &\stackrel{\Delta}{=} [\Theta_1,\Theta_2,\ldots,\Theta_{2L}]^T \nonumber \\
  &= [\bm{\theta}^{T},\bm{ \phi}^{T}]^{T} = [\theta_1,\ldots,\theta_L,\phi_1,\ldots,\phi_L]^T.
  \label{equ:ParamDeter}
  \end{align}
  According to the CRB theory,  the variance of $r$'th entry of unbiased estimator  ${\bf  \hat{\Theta}}$    satisfies the  following inequality:
  \begin{equation}
  \mathrm{var}({\bf  \hat{\Theta}}_r) \geq [{\bf F}({\bf  {\Theta}})]_{rr}^{-1},\quad 1\leq r \leq 2L,
  \label{equ:CRBoundn}
  \end{equation}
  where the element of $(r,s)$ of Fisher information matrix ${\bf F}({\bf  {\Theta}})$ is defined as:
  \begin{equation}
  [{\bf F}({\bf  {\Theta}})]_{rs} \stackrel{\Delta}{=} -\mathrm{E}\left\{{{{\partial ^2}\ln  f({{\bm b}};{\bf \Theta}) \over {\partial {\Theta _r}\partial {\Theta _s}}}} \right\}, \quad 1\leq r,s \leq 2L,
  \label{equ:FisherInfo}
  \end{equation}
  and  the density function of the observation 
  $f({{\bm b}};{\bf \Theta})$
  will be
  %             Statistical distribution of observations in ${\bf F}({\bf \Theta})$  can be rewritten as below:
  \begin{align}
  f({{\bm b}};{\bf \Theta}) = &{1 \over {{(2\pi) ^{P N_s/2}}\vert\det({{\bf {C}}_{b}})\vert^{N_s/2}}} 
  \nonumber \\
  &\times\exp\left(- \frac{1}{2}\sum_{t=1}^{N_s}({{{\bm b}}(t)}-\bm {\eta}(t))^T {\bf{C}}_{b}^{- 1}({{{\bm b}}(t)}-\bm {\eta}(t))\right)
  \label{equ:TotProbDen}
  \end{align}
  where  
  $  \bm {\eta}(t)=\mathrm{E}\{{{\bm  b}}({\bf \Theta};t)\} = {{\bf Y}}({\bf \Theta})^T{\bm s}(t)$  and 
  ${\bf C}_b = \mathrm{cov}({\bm b}(t)) = {\bf Q}$.
  After some mathematical manipulation, 
  the derivative of the density function with respect to  ${\Theta _r}$  and ${\Theta _s}$ can be simplified as  follows:
  %                Now, we calculate the derivative of the observation distribution function 
  \begin{align}
  {{{\partial \ln f({{\bm b}};{\bf \Theta})} \over {\partial{ \Theta}_s}}} &=\sum_{t=1}^{N_s}  \mathrm{tr}\left\{ 
  {\partial \bm {\eta}(t) \over \partial { \Theta}_s} \left({{{\bm b}}(t)}-\bm {\eta}(t)\right)^T   {\bf{C}}_{b}^{- 1}		\right\} \\
  %               \end{align}
  %               \begin{align}
  {{{\partial^2 \ln f({{\bm b}};{\bf \Theta})} \over {\partial{ \Theta}_r}{\partial{ \Theta}_s}}} &= \sum_{t=1}^{N_s} \mathrm{tr}\left\{  \left({\partial^2 \bm {\eta}(t) \over\partial { \Theta}_r  \partial { \Theta}_s}\left({{{\bm b}}(t)}-\bm {\eta}(t)\right)^T \right. \right. \nonumber \\ &\qquad\,\,\, \qquad- 
  \left. \left. {\partial \bm {\eta}(t) \over \partial { \Theta}_s } {\partial \bm {\eta}(t) \over \partial { \Theta}_r }^T    \right)
  {\bf{C}}_{b}^{- 1} \right\} .
  \end{align}
  Therefore (\ref{equ:FisherInfo}) can be rewritten as
  %               Applying the expectation operator to both sides of above equation, entries of Fisher information matrix are derived:
  \begin{equation}
  [{\bf F}({\bf \Theta})]_{rs} =  E\left\{{{{\partial^2 \ln f({{\bm b}};{\bf \Theta})} \over {\partial{ \Theta}_r}{\partial{ \Theta}_s}}} \right\}	= - \sum_{t=1}^{N_s}  {\partial \bm {\eta}(t) \over \partial { \Theta}_r }^T  {\bf{C}}_{b}^{- 1} {\partial \bm {\eta}(t) \over \partial { \Theta}_s },
  \label{equ:FishInfon}
  \end{equation}
  where
  \begin{equation}
  {\partial \bm {\eta}(t) \over \partial { \Theta}_r } =  {\partial {{\bf Y}}({\bf \Theta}) ^T\over \partial  { \Theta}_r}  {\bm s}(t)= \dot{{{\bf Y}}}_{\Theta_r}({\bf \Theta})^T
  {\bm s}(t).
  \label{equ:DiffEttan}
  \end{equation}
  The $r$'th column of  ${\bf Y}({\bf\Theta})$  is the function of $(\theta_r,\phi_r)$.
  %                   $(\theta_r,\phi_r)$ exists only in the $r$'th column of    ${\bf Y}({\bf\Theta})$  matrix. 
  Thus,  the derivative of   ${\bf Y}({\bf\Theta})$   respect to 
  $\theta_r$ or $\phi_r$  yields a  matrix with all zero elements except the 
  $r$'th  column.
  %                   
  %                   one of these two parameters, all of its columns will be zero except of the $r$'th one. 
    The derivative of matrix   ${\bf Y}({\bf\Theta})^،$  respect to  $\bm \theta$ and $\bm \phi$ are defined as follows:                   
  \begin{equation}
  {\dot {\bf Y}}_{{\bm \theta}}^T \stackrel{\Delta}{=} \sum\limits_{r = 1}^L {\dot {\bf Y}}_{{\theta _r}}^T, \quad
  {\dot {\bf Y}}_{{\bm \phi}}^T  \stackrel{\Delta}{=}  \sum\limits_{r = 1}^L {\dot {\bf Y}}_{{\phi _r}}^T,
  \label{equ:VectorDerivative1}
  \end{equation}
  where the  scalar derivatives ${\dot {\bf Y}}_{{\theta _r}}^T$  and   ${\dot {\bf Y}}_{{\phi _r}}^T$ 
  are respect to ${{\theta _r}}$ and ${{\phi _r}}$, respectively. The reverse equation of (\ref{equ:VectorDerivative1}) can be expressed as:
  \begin{equation}
  {\dot {\bf Y}}_{{\theta _r}}^T= {\dot {\bf Y}}_{{\bm \theta}}^T{e_r}e_r^T, \quad
  {\dot {\bf Y}}_{{\phi _r}}^T = {\dot {\bf Y}}_{{\bm \phi}}^T{e_r}e_r^T,
  \label{equ:ScalarDriv}
  \end{equation}
  where the  vector $e_r$    is the $r$'th column of the  identity matrix   ${\bf I}_r$.  According to  (\ref{equ:ParamDeter}),
  the  Fisher information matrix can be declared with the block matrix as follows:
  \begin{equation}
  {\bf F} =  \begin{bmatrix}
  {\bf F}_{{\bm \theta,\bm\theta}} & {\bf F}_{{\bm \theta,\bm\phi}} \\
  {\bf F}_{{\bm \phi,\bm\theta}} & {\bf F}_{{\bm \phi,\bm\phi}}
  \end{bmatrix}.
  \label{equ:FisherMat3n}
  \end{equation}
  where   ${\bf F}_{{\bm \theta,\bm\phi}} $  is a $L \times L$ matrix and 
  $[{\bf F}_{{\bm \theta,\bm\phi}} ]_{rs}$ 
  is obtained the same as    (\ref{equ:FishInfon}) and  the first and second derivatives are taken with respect to $r$'th and $s$'th entry of $\bm{\theta}$  and  $\bm{ \phi}$, respectively.
  %                     one of the entries of vector 
  Using (\ref{equ:FishInfon})-(\ref{equ:ScalarDriv}), 
  $[{\bf F}_{{\bm \theta,\bm\phi}} ]_{rs}$  can be simplified as
  \begin{align}
  [{\bf F}_{{\bm \theta,\bm\phi}} ]_{rs} &= {{\sum}}_{t=1}^{N_s}  \left(\dot{{{\bf Y}}}_{  \theta_r}^T  {\bm s}(t) \right)^T  {\bf C}_{b}^{- 1} \left(\dot{{{\bf Y}}}_{\phi_s}^T {\bm s}(t) \right)   \nonumber \\
  &= {{\sum}}_{t=1}^{N_s}  \left(  {\dot {\bf Y}}_{{\bm \theta}}^T{e_r}e_r^T  {\bm s}(t) \right)^T  {\bf C}_{b}^{- 1} \left(  {\dot {\bf Y}}_{{\bm \phi}}^T{e_s}e_s^T  {\bm s}(t) \right)  \nonumber \\
  &= {{\sum}}_{t=1}^{N_s} {\bm s}(t)^T {e_r} e_r^T  {\dot {\bf Y}}_{{\bm \theta}} {\bf C}_{b}^{- 1}    {\dot {\bf Y}}_{{\bm \phi}}^T{e_s}e_s^T  {\bm s}(t)   \nonumber \\
  &=  {{\sum}}_{t=1}^{N_s}  e_s^T  {\bm s}(t)  {\bm s}(t)^T {e_r} e_r^T  {\dot {\bf Y}}_{{\bm \theta}} {\bf C}_{b}^{- 1}    {\dot {\bf Y}}_{{\bm \phi}}^T {e_s}.
  \label{equ:FishInfonn}
  \end{align}
  Eventually defining    ${\bf S}_s= \sum_{t=1}^{N_s}   {\bm s}(t)  {\bm s}(t)^T  $, the matrix  ${\bf F}_{\bm{\theta},\bm{\phi}}$  will be:
  \begin{equation}
  {\bf F}_{\bm{\theta},\bm{\phi}} =  {\bf S}_s   \odot \left(  {\dot {\bf Y}}_{{\bm \theta}} {\bf C}_{b}^{- 1}    {\dot {\bf Y}}_{{\bm \phi}}^T  \right), 
  \end{equation}
  where $\odot$   represents the Hadamard product and is defined for two matrices as:
  \begin{equation}
  {[\bf A \odot \bf B]_{rs}}\triangleq{[\bf A]_{rs}}{[\bf B]_{rs}}.
  \end{equation}
It must be noted that  matrices 
  ${\bf F}_{{\bm \theta,\bm\theta}} $,  ${\bf F}_{{\bm \phi,\bm\theta}}$  and  ${\bf F}_{{\bm \phi,\bm\phi}}$
  are similarly defined. 
  %                    Using   (\ref{equ:FishInfon}) and  (\ref{equ:DiffEttan}),    $(r,s)$ entry for each block of Fisher information in
  %                    (\ref{equ:FisherMat3n})
  %                     matrix will be:
  %                                        Now assume that    $\Theta_r$  and  $\Theta_s$ are an entry of vectors  $\bm \theta$  and $\bm \phi$ , respectively. 
  %                    Other matrix blocks in (\ref{equ:FisherMat3n}) are derived similarly. For instance,    $	{F_{{\bm \theta}, {\bm \theta}}}$ is:
  %                      \begin{equation}
  %                      {\bf F}_{\bm{\theta},\bm{\phi}} = {\bf S}_s  \odot \left(  {\dot {\bf Y}}_{{\bm \Theta}} \bm{C}_{b}^{- 1}    {\dot {\bf Y}}_{{\bm \Theta}}^T  \right).
  %                      \end{equation}
  
  Considering the algebraic equality 
  %                      Finding the Fisher information matrix and using below matrix equality:
  \begin{equation}
  \begin{bmatrix}
  {\bf A}_{11} & {\bf A}_{12} \\
  {\bf A}_{21} & {\bf A}_{22}
  \end{bmatrix}^{-1} = 
  \begin{bmatrix}
  {\bf C}_1^{-1} & -{\bf A}_{11}^{-1}{\bf A}_{12}{\bf C}_2^{-1} \\ 
  -{\bf C}_2^{-1}{\bf A}_{21}{\bf A}_{11}^{-1} & {\bf C}_2^{-1}
  \end{bmatrix} ,
  \end{equation}
  where ${\bf A}_{11} $, ${\bf A}_{12} $, ${\bf A}_{21} $ and ${\bf A}_{22} $ are $L \times L$ matrix and 
  $
  {\bf C}_1 = {\bf A}_{11} - {\bf A}_{12} {\bf A}_{22}^{-1}{\bf A}_{21}$ and ${\bf C}_2 = {\bf A}_{22} -{\bf A}_{21}{\bf A}_{11}^{-1}{\bf A}_{12}
  $.
  Consequently, the CRB of ${\bf  \hat{\Theta}}$ is found using (\ref{equ:CRBoundn}):
  \begin{align}
  \mathrm{var}(\theta_l) &\geq  [C_1]_{ll}, l=1,\cdots,L,\\
  \mathrm{var}(\phi_l) &\geq  [C_2]_{ll}, l=1,\cdots,L,
  \end{align}

\section*{Acknowledgements}
The authors would like to acknowledge Nicolas Epain and Andrew Wabnitz   from CARLab,   the university of Sydney, Australia for providing  us the HOA and MCRoomSim toolbox.
\bibliographystyle{IEEEtran}

\end{document}